%% file: FC2027/main-FC-short.tex
\documentclass[runningheads]{llncs}
\usepackage[T1]{fontenc}
\usepackage{graphicx} % Required for inserting images
\usepackage{url,amsmath,amssymb,color}
\usepackage{float}
\usepackage{soul}
\usepackage[colorinlistoftodos]{todonotes}
\usepackage{tikz}
\usepackage{pgfplots}
\pgfplotsset{compat=1.18}
\usepackage{booktabs}
\usepackage{tcolorbox}
\usepackage[hidelinks]{hyperref}

\newtheorem{assumption}{Assumption}

\usepackage{graphicx}
\begin{document}
\title{Rewarding Informed Voting in PoS Governance\thanks{We would like to thank Paolo Penna, members of the Cardano Foundation and anonymous reviewers for helpful comments.}}%tentative title
\titlerunning{Rewarding Informed Voting in PoS Governance}
\author{Juan Beccuti \and
Manvir Schneider}
\authorrunning{J. Beccuti and M. Schneider}
 \institute{Cardano Foundation, Switzerland \\
 \email{\{firstname.lastname\}@cardanofoundation.org}}
% Springer Heidelberg, Europaplatz 3, 69115 Heidelberg, Germany
% ABC Institute, Ruprecht-Karls-University Heidelberg, Heidelberg, Germany\\
\maketitle              % typeset the header of the contribution
\begin{abstract}
We study reward design for delegated representatives who acquire costly information and vote on a binary governance proposal. Effort improves both private information and delegated voting weight. We show that no payment rule can make sincere voting strictly optimal at every positive effort profile under a non-neutral prior, while fixed budgets limit accuracy in large symmetric committees. We then compare proportional sharing with a conditional rule that rewards only members of a sufficiently large winning coalition. Unlike proportional sharing, which supports sincere voting only through indifference, the conditional rule makes sincere voting strictly optimal at symmetric positive effort profiles. It supports sincere equilibria with the same effort as proportional sharing and, under suitable reward-cost conditions, with strictly higher effort and decision accuracy. Insincere equilibria may nevertheless coexist.

\end{abstract}
\input{sections/intro-short}
\input{sections/literature-short}
\input{sections/model-short}

\input{sections/psr-short}

\input{sections/cpsr-short}
\input{sections/conclusion-short}

\bibliographystyle{splncs04}
\bibliography{BiblioGov}

\setcounter{section}{0}
\renewcommand{\thesection}{\Alph{section}}

\input{sections/appendix-proof-short}

\input{sections/appendix-results-examples-short}
\input{sections/appendix-implementation}

\end{document}

%% file: sections/intro-short.tex
\section{Introduction}\label{sec: intro}

Collective governance faces a familiar participation problem: voting and
evaluating proposals are costly, while an individual stakeholder may
have little chance of influencing the collective decision. Even
stakeholders who care about the outcome may therefore have weak
incentives to participate or become informed. In blockchain governance,
this difficulty is compounded by the technical and economic judgment
required to assess proposals concerning protocol changes and community
resources. Widely distributed voting rights do not, by themselves,
ensure informed participation. Indeed, empirical evidence from major
decentralized autonomous organizations documents limited voting
participation and concentrated decision-making power
\cite{falk2024blockchain}.

Delegation offers one response by allowing stakeholders to entrust their
voting power to representatives \cite{bongaerts2025vote}. Cardano provides
a concrete example: ada holders can delegate their voting power to
delegated representatives (DReps), who vote on governance actions
\cite{Decentrilized_Gov_CIP-1694}. Delegation reduces the need for every stakeholder to
evaluate every proposal, but sustaining representatives' work still
requires well-designed incentives. Evaluating proposals and communicating
assessments to delegators impose individual costs, while compensation
may come from a limited reward budget shared across the committee.
How that budget is distributed therefore shapes representatives'
incentives to become informed and vote accordingly.

Motivated by this setting, we study a committee of DReps deciding
whether to accept or reject a proposal of uncertain quality. Each DRep
chooses costly effort, which improves her private information and
attracts greater delegation support, increasing her voting weight.
After observing the distribution of voting power and updating her
belief using her private signal, she votes simultaneously with the
other DReps. The collective decision is determined by weighted majority
voting. DReps care only about their payments and effort costs, without
an intrinsic preference over the outcome. The designer therefore seeks
a reward rule that encourages both information acquisition and sincere
voting, meaning voting for the alternative each DRep believes is more
likely to be correct.

We first establish two general limitations.

\begin{tcolorbox}[
  colback=white,
  colframe=black,
  boxrule=0.4pt,
  arc=0pt
]
\textbf{Result 1. General limits.}

When the prior favors one state, no payment rule can satisfy strict
sincere-voting incentive compatibility (strict SVIC, Def.~\ref{pr:requirements}) at every positive
effort profile.
% Moreover, under a fixed reward budget and participation constraints, no positive accuracy advantage over choosing the prior-favored alternative can remain bounded away from zero as the committee grows along symmetric equilibria.
With a fixed reward budget and linear effort costs, any positive
accuracy advantage over choosing the prior-favored alternative
vanishes as committee size grows along symmetric equilibria
satisfying participation constraints.

\smallskip
\textit{Theorems~\ref{pr:strict-impossibility}
and~\ref{pr:accuracy}.}
\end{tcolorbox}

The first limitation leaves open whether strict SVIC can hold at
equilibrium. The second follows from the budget's restriction on
aggregate effort and bounds the accuracy attainable in large symmetric
committees.

We then compare the proportional sharing rule (PSR), which rewards
voting power independently of votes, with the conditional proportional
sharing rule (CPSR), which distributes rewards proportionally within
the winning coalition only if it contains a majority of DReps.

\begin{tcolorbox}[
  colback=white,
  colframe=black,
  boxrule=0.4pt,
  arc=0pt
]
\textbf{Result 2. Strict sincerity at the voting stage.}

PSR supports sincere voting only through indifference, whereas CPSR
satisfies strict SVIC at every symmetric positive effort profile.
For every interior prior, CPSR also supports a full symmetric sincere
equilibrium with the same effort and decision accuracy as sincere PSR.

\smallskip
\textit{Propositions~\ref{pr:PSR-properties}
and~\ref{pr:CPSR-properties};
Theorem~\ref{thm:cpsr-benchmark-sincere}.}
\end{tcolorbox}

CPSR makes a representative's vote relevant to payment eligibility.
This strengthens voting incentives while preserving the PSR effort
benchmark. We next ask whether CPSR can also support higher effort
and accuracy.

\begin{tcolorbox}[
  colback=white,
  colframe=black,
  boxrule=0.4pt,
  arc=0pt
]
\textbf{Result 3. Higher effort and accuracy.}

For every committee size and interior prior, CPSR supports symmetric
sincere equilibria with strictly higher effort and success probability
than sincere PSR when the reward-cost ratio is sufficiently close to,
but below, the level at which PSR induces maximal effort.

Once PSR induces maximal effort, CPSR supports the same effort and
perfect accuracy. These equilibria satisfy strict SVIC.

\smallskip
\textit{Propositions~\ref{prop:cpsr-interior-existence}
and~\ref{prop:maximal-effort};
Corollary~\ref{cor:cpsr-higher-success}.}
\end{tcolorbox}

The effort comparison reflects competing incentives: better information
can increase the probability of receiving a reward, while sharing
within the paid coalition provides a smaller marginal incentive to
increase one's payment share. The stated reward-cost conditions allow
the first effect to dominate. These are equilibrium existence results:
insincere pooling equilibria may coexist, so the improvements do not
establish uniqueness or guarantee equilibrium selection.

\paragraph{Roadmap.}
Section~\ref{section:model} introduces the model, incentive properties,
and equilibrium concept, and establishes general limitations on strict
sincere-voting incentives and decision accuracy.
Section~\ref{section:psr} characterizes the PSR benchmark.
Section~\ref{section:cpsr} analyzes CPSR's voting incentives and
equilibria, including conditions supporting higher effort and accuracy
than PSR.
Section~\ref{section:conclusions} concludes.
Appendix~\ref{appendix:proofs} contains the proofs of the main results.
Appendix~\ref{appendix:results-examples} provides additional analytical
results, worked examples, and supporting probability and payoff
calculations.
Appendix~\ref{app:implementation} outlines a possible implementation
path for Cardano.

%% file: sections/literature-short.tex
\subsection{Related Literature}\label{sec: literature}

Our paper connects research on blockchain governance
\cite{kiayias2023sok} with the economic literature on
information acquisition and strategic voting, and with
mechanisms for eliciting information from compensated agents.

In the first branch, our work is closest to
\cite{birmpas2024reward}, who study effort incentives and
committee size when information acquisition attracts
delegated voting power. Their analysis shows that proportional
rewards under a fixed budget can lead to vanishing effort
and near-random decisions in large committees. Whereas their
model directly links effort to the probability of a correct
vote, we distinguish acquiring information from acting on it:
representatives choose how to vote after updating their
beliefs, making sincere voting an additional incentive
requirement.

\cite{tsoukalas2020token} also study costly information
acquisition and strategic voting under token-weighted
aggregation. Their participants benefit from accurate
estimates through token values; our outcome-neutral
representatives require compensation to become informed.
The distinction between stake shares and voting power is
examined in \cite{zhang2026swimming}, while
\cite{gersbach2022risky,gersbach2023impact,gersbach2024effects}
study how delegation affects participation, voting power,
and collective outcomes. We take the relation between effort
and delegated voting weight as given and examine how rewards
jointly incentivize information acquisition and sincere voting.

The joint provision of these incentives connects our paper
to \cite{caragiannis2022outsourcing}, who study costly effort
and strategic voting by outcome-neutral jurors with equal
voting weights. They use rewards linked to majority agreement
to support target accuracy at minimum expected cost. Our
distinction lies in endogenous voting weights: information
acquisition changes both a representative's influence and
her compensation, affecting incentives to vote sincerely.
\cite{dasgupta2013crowdsourced} instead use evaluations
across multiple tasks to induce effort and truthful reporting
of private observations. We study a single collective
decision, where votes determine both the outcome and
reward eligibility.

Related compensation questions arise in
\cite{azrieli2021monitoring}, which characterizes minimum-cost
contracts for outcome-neutral experts acquiring costly
information, with optimal payments only when reports agree.
\cite{azrieli2022delegated} extends this analysis to general
peer-monitoring environments. These studies connect
compensation to information acquisition and truthful
reporting; we examine these incentives within delegated
weighted-majority voting, where effort determines voting
weight and compensation is restricted by a fixed budget.

Another branch, in economics and social choice, examines
information aggregation by voters who value collective
outcomes. \cite{austensmith1996information} establish that
common interests do not generally make sincere voting an
equilibrium. With endogenous information,
\cite{persico2004committee} studies how committee size and
voting thresholds shape acquisition incentives, while
\cite{martinelli2006would} shows that vanishing individual
information can coexist with accurate collective decisions
when weak information is sufficiently cheap. In these
baseline environments, information is privately valuable
through its effect on the collective decision. In our
setting, a vote can affect compensation even without
changing that decision.\footnote{\cite{persico2004committee}
also considers sufficiently large transfers rewarding
agreement with other members' votes; our analysis restricts
compensation to a fixed reward budget.}
The cost and financing assumptions also matter for
aggregation: with linear effort costs and a fixed budget,
our bound implies that any positive accuracy advantage
over the prior-based decision vanishes as committee size
grows along symmetric equilibria.

\cite{gersbach2022appointed} and \cite{gersbach2024we}
introduce monetary incentives into committee design through
incentive transfers and sharing of verifiable information
costs, respectively, obtaining conditions favoring small
committees. Their mechanisms supplement members' interest
in correct collective decisions. Our analysis instead
concerns representatives without that underlying interest,
for whom remuneration provides the entire incentive to
acquire information.

Finally, \cite{gerardi2008information} and
\cite{gershkov2009optimal} examine how the decision procedure
itself can encourage information acquisition. The former
choose committee size and the rule translating votes into
decisions; the latter characterize optimal sequential
acquisition and voting procedures. These approaches can
sacrifice efficient use of acquired information to improve
acquisition incentives. Our design question concerns
compensation within given institutional constraints:
committee size and simultaneous weighted-majority voting
are taken as given.

%% file: sections/model-short.tex
\section{Model}\label{section:model}

We consider a vote on a proposal submitted for the DReps' consideration. Proposal quality is uncertain ex ante. Let the true state be $S\in\{G,B\}$ with common prior $\mu \equiv P(G)\in(0,1)$ that the proposal is of good quality and $P(B)=1-P(G)$ that it is of bad quality.

There are $N=\{1,\dots,2n+1\}$ ex-ante identical \textit{delegated representatives} (from now on, DReps) participating in a binary governance decision to either accept or reject a proposal, where abstentions are abstracted away for future work. Each DRep \(i\) simultaneously chooses an effort level \(x_i\in[0,\frac{1}{2}]\), which is costly and given by the linear cost function \(c(x_i)=ax_i\), where \(a>0\). Let
\[\mathbf{x}=(x_1,\ldots,x_{2n+1}) \in \left[0,\frac12\right]^{2n+1}\]
denote the effort profile.

As in \cite{birmpas2024reward}, higher effort not only improves information quality but also increases the voting weight attached to that DRep. Accordingly, voting power is endogenous and given by
\[
w_i(\mathbf{x})
:=
\begin{cases}
\dfrac{x_i}{\sum_{j=1}^{2n+1}x_j},
& \text{if } \sum_{j=1}^{2n+1}x_j>0,\\[3mm]
0,
& \text{if } \sum_{j=1}^{2n+1}x_j=0,
\end{cases}
\]
The interpretation is that
higher effort implies a better-informed DRep, which increases the support delegated to her vote. After all DReps have chosen effort, the resulting
voting-power profile
\[
\mathbf{w}(\mathbf{x})
:=
\bigl(w_1(\mathbf{x}),\ldots,w_{2n+1}(\mathbf{x})\bigr)
\]
becomes publicly observable.

For any DRep \(i\) with \(x_i>0\), however, observing
\(\boldsymbol{w}\) together with her own effort is sufficient to recover
the entire effort profile, since
\[
\sum_j x_j=\frac{x_i}{w_i},
\qquad
x_j=\frac{w_j}{w_i}x_i
\quad\text{for every }j.
\]
Hence, for positive effort profiles, observing the voting-power profile is informationally equivalent to observing the effort profile itself. For notational convenience, we  write voting strategies and continuation payoffs as functions of \(\mathbf{x}\). Throughout the remainder of the paper, references to the effort profile available at the voting stage should be understood as referring to the effort profile inferred from the publicly observed voting-power profile.

After all DReps have chosen their effort, each DRep \(i\) then privately observes a signal \(s_i\in\{g,b\}\). The signal has symmetric precision across states:
\[ p_i
:=
\Pr(g\mid G,x_i) = \Pr(b\mid B,x_i) = \frac12+x_i.
\]
Conditional on the state and the effort profile, private
signals are independent across DReps. Higher effort raises the likelihood that DRep \(i\)'s signal is aligned with the true state. Upon observing \(s_i\), DRep \(i\) updates her belief about the state by Bayes' rule
\begin{align*}
    \hat{\mu}_{i,g}&:=\Pr(G\mid g,x_i)=\frac{\mu p_i}
{\mu p_i+(1-\mu)(1-p_i)},\\
\hat{\mu}_{i,b}&:=\Pr(G\mid b,x_i)=\frac{\mu(1-p_i)}
{\mu(1-p_i)+(1-\mu)p_i}.
\end{align*}

After inferring the effort profile \(\mathbf{x}\), receiving her private signal, and forming posterior belief \(\hat{\mu}_i\), DRep \(i\) chooses a vote
\[
    v_i(\hat{\mu}_i,\mathbf{x})\in\{0,1\},
\]
where \(v_i=1\) denotes acceptance and \(v_i=0\) denotes rejection. Thus, voting depends both on the DRep's posterior belief and on the distribution of effort inferred from the observed voting-power profile. 

A pure strategy for DRep \(i\) consists of an effort choice
\(x_i\in[0,\frac12]\) and a voting rule
\(v_i:[0,1]\times\left[0,\frac12\right]^{2n+1}\rightarrow
\{0,1\}\). Let \(\mathbf{v}=(v_1,\ldots,v_{2n+1})\) denote the profile of voting rules.

\paragraph{Timing and information structure.}
The timing and information structure of the governance game can be summarized as
\[
\underbrace{\mu}_{\substack{\text{common}\\\text{prior}}}
\;\longrightarrow\;
\underbrace{(x_i)_{i\in N}}_{\substack{\text{simultaneous}\\\text{effort choices}}}
\;\longrightarrow\;
\underbrace{\mathbf{x}}_{\substack{\text{inferred from}\\\text{publ. observed}\\ \mathbf{w}}}
\;\longrightarrow\;
\underbrace{s_i}_{\substack{\text{private}\\\text{signal}}}
\;\longrightarrow\;
\underbrace{\hat{\mu}_i}_{\text{posterior}}
\;\longrightarrow\;
\underbrace{v_i(\hat{\mu}_i,\mathbf{x})}_{\text{vote}},
\]
where arrows denote chronological sequence.

\begin{definition}[Sincere Voting]
For any inferred effort profile \(\mathbf{x}\), DRep \(i\) votes sincerely
if her vote depends only on her posterior belief and supports the
posterior-more-likely state. That is,
\[
v_i^{sv}(\hat{\mu}_i,\mathbf{x})
\equiv
v_i^{sv}(\hat{\mu}_i)
=
\mathbf{1}
\left\{
\hat{\mu}_i>\frac12
\right\}.
\]
\end{definition}

Using Bayes' rule, sincere voting implies
\[
v_i^{sv}(\hat{\mu}_{i,g},\mathbf{x})
=
\mathbf{1}
\left\{
\mu>\frac12-x_i
\right\},
\qquad
v_i^{sv}(\hat{\mu}_{i,b},\mathbf{x})
=
\mathbf{1}
\left\{
\mu>\frac12+x_i
\right\}.
\]

It follows that three cases may arise\footnote{The boundary cases follow from sincere rejection at posterior
$1/2$: at $\mu=1/2-x_i$, the posterior after $g$ equals
$1/2$, while at $\mu=1/2+x_i$, the posterior after $b$
equals $1/2$.}
\[
\begin{cases}
\mu\leq \frac12-x_i,
& \text{DRep } i \text{ rejects after both signals},\\[0.4em]
\frac12-x_i<\mu\leq\frac12+x_i,
& \text{DRep } i \text{ accepts after } g
  \text{ and rejects after } b,\\[0.4em]
\mu>\frac12+x_i,
& \text{DRep } i \text{ accepts after both signals}.
\end{cases}
\]

Notice that these prior ranges are endogenous to DRep \(i\)'s effort. Under sincere voting, a DRep’s effort can induce a shift in her decision rule, potentially moving her from a regime where she votes independently of her signal (SI) to a range where her vote becomes signal-contingent (SC). With \(v_i(\hat{\mu}_i,\mathbf{x})\), there is an additional strategic channel: the vote may also respond directly to the effort choices of the other DReps through the observed distribution of voting power.

%---------------------------------------------------------------------------------
\subsection{Expected Payoffs}\label{section:expected_payoffs}

For any subset \(C\subseteq N\), define its total voting weight by
\[
\omega(C;\mathbf{x})
:=
\sum_{j\in C}w_j(\mathbf{x}).
\]

Let $A\subseteq N$ be the set of DReps who vote to accept.
The proposal is accepted when
$\omega(A;\mathbf{x})>1/2$ and rejected when
$\omega(A;\mathbf{x})<1/2$. At a voting-weight tie,
an independent fair coin determines the outcome after votes
have been cast. The winning coalition is therefore
\[
W(A;\mathbf{x})=
\begin{cases}
A,
& \omega(A;\mathbf{x})>1/2,\\
N\setminus A,
& \omega(A;\mathbf{x})<1/2,\\
A \text{ or } N\setminus A,\ \text{each with probability }1/2,
& \omega(A;\mathbf{x})=1/2.
\end{cases}
\]

Let $f_i(W,\mathbf{x})$ denote the transfer received by
DRep $i$ when the realized winning coalition is $W$ and
the effort profile is $\mathbf{x}$. The payment rule
determines which DReps receive a transfer and its amount.
Since the protocol cannot observe the true state,
payments depend only on the realized winning coalition
and the effort profile, not on whether the state is $G$
or $B$.

Moreover, DReps do not derive idiosyncratic private benefits from one policy outcome or the other:

\begin{assumption}
DReps are outcome-neutral: they have no intrinsic preference for
acceptance or rejection beyond the payment rule and effort costs.
\end{assumption}

This assumption isolates the incentive effects of the reward rule from other motives that may shape DReps' behavior.\footnote{A related
assumption appears explicitly in
 \cite{azrieli2021monitoring} and \cite{george2023strategic}, where agents receive no intrinsic
utility from the outcome and compensation is the entire incentive.} Outcome neutrality does not imply indifference between
voting actions: a DRep's vote can affect which coalition
wins and, depending on the reward rule, whether and how
much she is paid.

Given an effort profile $\mathbf{x}$ and continuation
voting rules $v$, DRep $i$'s expected net payoff is
\begin{equation}
\pi_i(v,\mathbf{x};f)
=
\mathbb{E}\!\left[
f_i(W(A;\mathbf{x}),\mathbf{x})
\right]-c(x_i),
\label{eq:expected_payoff}
\end{equation}
where the expectation accounts for the common prior,
private signals, and the tie-breaking coin.
Appendix~\ref{app:probabilities-payoffs} provides the
probability formulas and expanded payoff expressions.

\subsection{Designer’s Problem and Incentive Properties}

The designer seeks a reward scheme that maximizes
$\Pr(\text{success};v,\mathbf{x})$, the probability
of accepting in state $G$ or rejecting in state $B$.
Appendix~\ref{app:probabilities-payoffs} gives its
expression, including the tie-breaking coin.

There are some incentive properties or constraints that the designer must account for when studying $f()$. 

\begin{definition}[Incentive properties]
\label{pr:requirements}
We consider the following requirements.
\begin{enumerate}
\item \emph{Effort-stage rationality (ER).}
Given continuation voting strategies $v^*$, each DRep's effort satisfies
\[
x_i^*\in\arg\max_{x_i\in[0,1/2]}
\pi_i\!\left(
v^*(\cdot;(x_i,\mathbf x_{-i}^*)),
(x_i,\mathbf x_{-i}^*);f
\right).
\]

\item \emph{Budget constraint (BC).}
For every effort profile and realized winning coalition,
\[
\sum_{i\in N}f_i(W,\mathbf x)\leq B.
\]

\item \emph{Participation constraint (PC).}
At the proposed equilibrium,
\[
\pi_i(v^*,\mathbf x^*;f)\geq0
\qquad\text{for every }i\in N.
\]

\item \emph{Sincere-voting incentive compatibility (SVIC).}
At a specified positive effort profile $\mathbf x$, for every DRep
and every signal having positive probability,
\[
v_i^{sv}(\hat\mu_{i,s_i})\in
\arg\max_{v_i\in\{0,1\}}
\mathbb E\!\left[
f_i(W(A;\mathbf x),\mathbf x)
\,\middle|\,
\hat\mu_{i,s_i},\mathbf x,v_i,v_{-i}^{sv}
\right].
\]
SVIC on a set of effort profiles means that this condition holds
at every profile in that set.

\item \emph{Strict SVIC.} \label{xxx}
SVIC holds, and the sincere action is the unique best response
whenever $\hat\mu_{i,s_i}\neq1/2$.
\end{enumerate}
\end{definition}

\begin{definition}[Equilibrium]
\label{pr:equilibrium}
A profile $(\mathbf x^*,v^*)$ is an equilibrium if it satisfies ER
and its feasible continuation strategies prescribe a best response
at every voting information set, including following effort deviations.
It is \emph{sincere on the equilibrium path} if the continuation
at $\mathbf x^*$ is $v^{sv}$.
\end{definition}

Thus SVIC verifies voting optimality where sincere voting is prescribed.
A sincere equilibrium additionally requires ER and rational
continuations after every effort deviation. PC follows from ER
under nonnegative transfers, since zero effort guarantees a
nonnegative payoff.

\begin{theorem}[A limitation of strict SVIC]
\label{pr:strict-impossibility}
For every $\mu\neq1/2$, no payment rule in this model satisfies
strict SVIC at every positive effort profile.
More specifically, if a rule satisfies SVIC at every positive profile,
then whenever
\[
0<x_j<|\mu-1/2|\quad(j\neq i),
\qquad
|\mu-1/2|<x_i\leq1/2,
\]
DRep $i$ must receive the same expected transfer from either vote
when the other DReps vote sincerely.
\end{theorem}

The limitation arises when the other DReps' sincere votes are independent of their signals, while DRep $i$'s sincere vote changes with hers. Sincere voting can be optimal for DRep $i$ after both signals only if the two actions yield equal expected payments, ruling out strict SVIC at these profiles. This prevents strict SVIC from holding at \emph{every} positive effort profile, but
does not preclude it at an equilibrium profile, such as the symmetric CPSR equilibria studied below.

\begin{theorem}[Budget and accuracy at symmetric equilibria]
\label{pr:accuracy}
For any payment rule satisfying BC, every equilibrium satisfying
PC obeys
\[
a\sum_{i\in N}x_i^*\leq B.
\]
If $\mathbf x^*=(x^*,\ldots,x^*)$, then, for any equilibrium
voting strategies,
\begin{equation}
\label{pr:accuracy-bound}
\Pr(\text{success};v^*,\mathbf x^*)
\leq
\min\left\{
1,\;
\max\{\mu,1-\mu\}
+2\min\{\mu,1-\mu\}\frac{B}{a}
  \frac{\binom{2n}{n}}{4^n}
\right\}.
\end{equation}
This bound does not require sincere voting.
In particular, for fixed $B$, $a$, and $\mu$, no such sequence
of symmetric equilibria can maintain a success probability
bounded strictly above $\max\{\mu,1-\mu\}$ as $n\to\infty$.
\end{theorem}

The first part bounds the aggregate effort that the
budget can support. At a symmetric equilibrium, it
implies $x^*\leq B/[a(2n+1)]$. The accuracy bound shows
that, under this restriction, increasing committee size
cannot preserve an accuracy advantage over choosing
the prior-favored alternative. The restriction to
symmetric effort is essential: the theorem does not
rule out concentrating the budget and information
acquisition on a smaller set of DReps.

%% file: sections/psr-short.tex
\section{Proportional Sharing Rule (PSR)}
\label{section:psr}

As a benchmark, consider a rule that distributes the budget in
proportion to effort, %independently of the collective decision and membership in the winning coalition:
\begin{equation}
\label{eq:proportional_sharing_rule}
f_i^{\mathrm{psr}}(\mathbf{x})
=
\begin{cases}
\dfrac{x_i}{\sum_{j\in N}x_j}B,
& \text{if } \sum_{j\in N}x_j>0,\\[1em]
0,
& \text{otherwise}.
\end{cases}
\end{equation}
Under PSR, $f_i(W,\mathbf{x})=f_i^{\mathrm{psr}}(\mathbf{x})$
for every winning coalition $W$: payment depends only on
effort.

Once effort has been chosen, payment is independent of the
DRep's posterior belief and of every DRep's vote. Consequently,
\[
\arg\max_{v_i\in\{0,1\}}
\mathbb{E}\!\left[
f_i^{\mathrm{psr}}(\mathbf{x})
\mid
\hat{\mu}_i,\mathbf{x},v_i,v_{-i}
\right]
=
\{0,1\}.
\]
Both voting actions are therefore optimal at every voting
continuation. At the effort stage, each DRep trades the increase
in her share of the budget against the cost of additional effort.
These incentives yield the following benchmark.

\begin{proposition}[Voting incentives and symmetric equilibrium under PSR]
\label{pr:PSR-properties}
PSR satisfies BC and SVIC at every positive effort profile.
It does not satisfy strict SVIC at profiles where some
posterior differs from $1/2$.
For every prior, it admits a symmetric sincere equilibrium
satisfying ER and PC, with
\[
x^{psr}
=
\min\left\{
\frac12,\frac{2nB}{a(2n+1)^2}
\right\}.
\]
\end{proposition}

Proposition~\ref{pr:PSR-properties} establishes that PSR
supports effort and sincere voting together, but only through
indifference at the voting stage: DReps receive the same
payment regardless of their votes. Non-sincere voting
continuations therefore support the same equilibrium effort.
PSR satisfies SVIC, but not strict SVIC.

Sincere PSR also provides a benchmark for decision accuracy.
At $\mu=1/2$, any accuracy improvement achieved by another
reward rule with the same budget and effort costs becomes
small compared with PSR's own advantage over random choice
as the committee grows, restricting attention to symmetric
equilibria. PSR's advantage nevertheless vanishes.
Corollary~\ref{cor:psr-asymptotic-accuracy} in
Appendix~\ref{app:additional-analytical-results}
states this result formally.

These findings leave two questions open: whether conditioning
payments on the voting outcome can induce strict SVIC, and
whether it can support higher effort and accuracy than
sincere PSR. The asymptotic comparison does not rule out
such improvements in finite committees and does not
establish the same conclusion for priors $\mu\neq1/2$.

%% file: sections/cpsr-short.tex
\section{Conditional Proportional Sharing Rule (CPSR)}
\label{section:cpsr}

CPSR conditions payment on membership in the winning coalition. In particular, DRep $i$ receives
\begin{equation}
\label{eq:conditional_payment}
f_i^{cpsr}(W,\mathbf{x})=
\begin{cases}
\displaystyle B\frac{x_i}{\sum_{j\in W}x_j},
&i\in W,\quad |W|\geq n+1,\quad \sum_{j\in W}x_j>0,\\[1mm]
0,&\text{otherwise}.
\end{cases}
\end{equation}
Notice that the protocol can implement these shares using observed voting weights,
since $w_i/\sum_{j\in W}w_j=x_i/\sum_{j\in W}x_j$.

The two payment requirements have different roles. Voting weight determines
which alternative wins; coalition size determines whether that winning
coalition is paid. A DRep can therefore determine the outcome through
sufficient voting weight without receiving a payment as a singleton.
For example, consider five DReps with weight vector $(0.6, 0.1, 0.1, 0.1, 0.1)$. DRep $1$ can make acceptance win by herself, but if the remaining four DReps vote to reject, the resulting winning coalition contains fewer than three members and consequently receives no payment. Without the
cardinality requirement, such a singleton winner would receive the entire budget.

This payment structure creates two margins through which effort
affects expected payment. The \emph{share margin} concerns the
proportion of the reward received when a DRep is paid. The
\emph{eligibility margin} concerns whether she receives a payment
at all. Under PSR, every positive-effort DRep receives a share
of the budget independently of voting. Under CPSR, effort also
improves private information and changes voting weight,
potentially affecting the probability of belonging to a
payment-eligible winning coalition. These additional incentives
depend on the voting continuation.

Appendix~\ref{app:probabilities-payoffs} derives the
expected payoff when one DRep chooses her effort and
voting rule while the others exert common effort and
vote sincerely. After effort deviations, this expression
applies only when the others continue to vote sincerely;
full equilibrium also requires verifying effort incentives
under sequentially rational voting continuations.

Example~\ref{ex:cardinality-requirement} (see Appendix \ref{app:examples}) illustrates how the cardinality requirement prevents a DRep from capturing the entire budget by acquiring enough voting weight to win alone.

\subsection{Analysis}

We first examine voting incentives and then establish effort equilibria. Unlike PSR, CPSR can make sincere voting uniquely optimal whenever the posterior differs from $1/2$, whether sincere votes follow signals or pool. Establishing a symmetric
equilibrium with strict SVIC additionally requires verifying ER under sequentially rational voting continuations.

The cardinality requirement makes pooling a credible
continuation: a unilateral dissenter receives no payment, even if her voting weight changes the outcome. The following proposition also establishes strict SVIC at every symmetric positive effort profile. These voting continuations provide
the basis for the equilibrium constructions below.

\begin{proposition}[Voting incentives and pooling equilibria under CPSR]
\label{pr:CPSR-properties}
Suppose $n\geq1$, $B>0$, and $a>0$.
\begin{enumerate}
\item CPSR satisfies BC at every effort profile.
\item After every effort profile, pooling on acceptance and pooling on
rejection are voting-stage equilibria, independently of the prior.
\item For every $\mu\in(0,1)$, CPSR satisfies strict SVIC at every
symmetric positive profile $(x,\ldots,x)$, $0<x\leq1/2$.
It need not satisfy SVIC at asymmetric positive profiles.
\item For every prior $\mu\in(0,1)$, CPSR admits two symmetric
pooling equilibria satisfying ER and PC with effort
\[
x^{psr}=\min\left\{\frac12,\frac{2nB}{a(2n+1)^2}\right\}:
\]
one pools on rejection and the other on acceptance. These equilibria
are sincere and satisfy strict SVIC when $\mu\leq1/2-x^{psr}$
and $\mu>1/2+x^{psr}$, respectively, under the convention that
sincere voting rejects at posterior $1/2$.
\end{enumerate}
\end{proposition}

The proposition shows that strict incentives for sincere
voting can coexist with multiple voting equilibria, including insincere equilibria. At a symmetric positive effort $x$, sincere voting follows signals when $|\mu-1/2|<x$. In this region, it coexists with pooling voting-stage equilibria that are insincere after one of the signals. For sufficiently extreme priors, pooling on the prior-favored alternative is itself sincere and satisfies strict SVIC, while pooling on the opposite alternative remains insincere. More generally, every CPSR equilibrium with positive total effort has a counterpart obtained by reversing all continuation voting rules, preserving effort and expected payoffs while replacing success probability by its complement.

At effort $x^{psr}$, both pooling outcomes can also be
sustained as full symmetric equilibria. Rejection is
sincere when $\mu\leq1/2-x^{psr}$, and acceptance when
$\mu>1/2+x^{psr}$. Supporting a full equilibrium with
sincere, signal-contingent voting additionally requires
verifying effort incentives under the chosen continuations.

Appendix~\ref{app:cpsr-asymmetric} illustrates voting
incentives across priors at the symmetric effort profile
$x=x^{psr}$ (see Figure~\ref{fig:cpsr-prior-regions-n1}) and provides an asymmetric-profile example
where CPSR fails SVIC (Example~\ref{pr:SVIC-example}).

Proposition~\ref{pr:CPSR-properties} establishes full symmetric pooling equilibria at effort $x^{psr}$, which are sincere and satisfy strict SVIC under sufficiently extreme priors. The
following theorem extends this conclusion to every interior prior by prescribing sincere voting at the symmetric effort profile and pooling on rejection after any other effort profile.
These continuations are voting-stage equilibria, and the proof shows that no DRep can profit from an effort deviation, even when choosing her best subsequent vote. Sincerity is required on the equilibrium path, but not after deviations.

\begin{theorem}[An equilibrium with strict SVIC for every prior]
\label{thm:cpsr-benchmark-sincere}
Suppose $n\geq1$, $B>0$, and $a>0$. For every prior
$\mu\in(0,1)$, CPSR admits a symmetric equilibrium satisfying
ER, strict SVIC, BC, and PC, in which every DRep chooses
\[
x^*=x^{psr}
=\min\left\{\frac12,\frac{2nB}{a(2n+1)^2}\right\}
\]
and votes sincerely on the equilibrium path. All voting
continuations are sequentially rational.
\end{theorem}

The theorem establishes that CPSR can support the same effort and decision accuracy as sincere PSR while making sincere voting strictly optimal whenever the posterior $\hat{\mu}\neq 1/2$.
This conclusion holds for every interior prior, whether sincere voting is signal-independent or signal-contingent. Other equilibria remain possible, so the result does not guarantee that this
equilibrium is selected. The remaining question is whether CPSR can also support higher effort and accuracy than PSR. We next examine effort incentives beyond $x^{psr}$.

\subsection{Effort incentives beyond the benchmark}

An effort deviation can change both the deviator's voting incentives
and those of the remaining DReps. The next lemma identifies sincere
continuations that can be used when studying efforts different from
$x^{psr}$.

\begin{lemma}[Strict voting incentives following unilateral effort deviations]
\label{lem:sincere_voting_effort_deviation}
Suppose $|\mu-1/2|<x^*\leq1/2$, and all DReps other than $i$ choose effort $x^*$.
\begin{enumerate}
\item If the other DReps vote sincerely, DRep $i$ satisfies the individual incentive requirement of strict SVIC at every $x_i\in(0,1/2]$: sincere voting is a best response and is
uniquely optimal whenever $\hat\mu_i\neq1/2$.
At $x_i=0$, both votes yield zero payment.

\item If $|\mu-1/2|<x_i\leq1/2$, sincere voting is
the unique best response for every DRep at each signal
having positive probability when the other DReps vote
sincerely. Hence CPSR satisfies strict SVIC at
$(x_i,x^*,\ldots,x^*)$, and sincere voting is a
voting-stage equilibrium there.
\end{enumerate}
\end{lemma}

The key is how payments compare across the two states. Under the lemma's conditions, acceptance gives DRep $i$ a higher expected payment in state $G$, while rejection gives an equally large advantage in state $B$. Her posterior determines how she weighs these two possibilities, so supporting the more likely state maximizes her expected payment, strictly whenever her effort is positive and her posterior differs from $1/2$.
Part~2 establishes that the other DReps retain this incentive when the deviator's sincere vote also follows her signal.

Example~\ref{ex:interior_SC_equilibria} compares effort
under CPSR and PSR conditional on sincere,
signal-contingent voting. The comparison can change sign
with committee size.
Example~\ref{ex:cpsr-sincere-equilibria} then combines
sincere and pooling voting continuations to establish
full symmetric equilibria satisfying ER, strict SVIC,
BC, and PC over suitable prior ranges.
Both examples and their supporting calculations appear
in Appendix~\ref{app:examples}.

\subsection{Higher effort and decision accuracy}

We now identify conditions under which CPSR supports
strictly higher effort and decision accuracy than
sincere PSR for every committee size and every
interior prior.

\begin{proposition}[Symmetric sincere equilibria with higher interior effort]
\label{prop:cpsr-interior-existence}
Suppose $n\geq1$, $B>0$, $a>0$, and $\mu\in(0,1)$, and let
\[
R_n:=\frac{(2n+1)^2}{4n},
\]
be the reward-cost ratio at which PSR first induces maximal effort, \(x^{psr}=1/2\). 
For every $n$ and $\mu$, there exists
$\varepsilon_n(\mu)\in(0,R_n)$ such that, whenever
\[
R_n-\varepsilon_n(\mu)<\frac{B}{a}<R_n,
\]
CPSR supports a symmetric equilibrium satisfying ER,
strict SVIC, BC, and PC, with effort $x^*$ satisfying
\[
\left|\mu-\frac12\right|<x^{psr}<x^*<\frac12,
\qquad
x^{psr}=\frac{2nB}{a(2n+1)^2}.
\]
Voting is sincere and signal-contingent on the equilibrium
path, and all voting continuations are sequentially rational.
\end{proposition}

The proposition establishes that, for every committee size
and every interior prior, CPSR can sustain strictly higher
effort than PSR under the stated reward-cost conditions.
The improvement holds at the same budget and effort cost as PSR. 
Its construction combines sincere voting after some effort deviations with pooling after others, with every prescribed continuation being a voting-stage equilibrium. Thus, the result establishes what CPSR can support under suitable continuation selection, rather than predicting which equilibrium participants will select. The reward-cost interval is nonempty, but the result does not quantify its width or provide a
uniform lower bound on the accuracy improvement.

The required reward-cost conditions depend on committee
size. These improvements therefore do not remove the
fixed-budget limitation established in
Theorem~\ref{pr:accuracy}.

The inequality $|\mu-1/2|<x^{psr}<x^*$ ensures that sincere
voting follows signals under both schemes. Thus, higher effort
under CPSR improves signal accuracy and raises
the probability that the majority identifies the state
correctly. The next corollary states this implication.

\begin{corollary}[Higher success probability under sincere CPSR]
\label{cor:cpsr-higher-success}
Under the conditions of Proposition~\ref{prop:cpsr-interior-existence},
the CPSR equilibrium constructed in it with strict SVIC has strictly
greater success probability than the sincere PSR equilibrium:
\[
\Pr(\text{success};v^{sv},\mathbf{x}^*)>
\Pr(\text{success};v^{sv},\mathbf{x}^{psr}),
\]
where $\mathbf{x}^*=(x^*,\ldots,x^*)$ and
$\mathbf{x}^{psr}=(x^{psr},\ldots,x^{psr})$.
\end{corollary}

\paragraph{Illustration.}
Consider three DReps, $B=1.2$, $a=1$, and $\mu=1/2$.
The equilibrium constructed in Examples~\ref{ex:interior_SC_equilibria} and \ref{ex:cpsr-sincere-equilibria} (which use the same continuation structure as Proposition~\ref{prop:cpsr-interior-existence}, see Appendix~\ref{app:examples}) has effort
$x^*\approx0.31743$, compared with $x^{\mathrm{psr}}\approx0.26667$. Its success probability is approximately $0.91217$, compared with $0.86207$ under sincere PSR. Each DRep's expected gross payment is $B/3=0.4$ under both equilibria, so expected net payoffs are approximately $0.08257$ and $0.13333$, respectively. Thus, the more accurate equilibrium
requires greater effort and gives each DRep a lower expected net payoff. CPSR also supports a sincere equilibrium at the PSR effort level.

The strict accuracy comparison relies on sincere voting being
signal-contingent under both schemes. When the prior is
sufficiently extreme that sincere voting instead pools on
the prior-favored alternative under both schemes, their
success probability is the same, namely $\max\{\mu,1-\mu\}$.
In particular, this applies to the sincere pooling equilibria
at effort $x^{psr}$ established above. CPSR still provides
strict SVIC at those profiles, whereas PSR supports sincere
voting through indifference.

Once $B/a\geq(2n+1)^2/(4n)$, PSR already induces
maximal effort. CPSR supports the same effort and perfect
accuracy for every interior prior, together with strict
SVIC. Proposition~\ref{prop:maximal-effort} in the appendix
states this result formally.

%% file: sections/conclusion-short.tex
\section{Conclusion and Discussion}\label{section:conclusions}

This paper studies how a limited reward budget can encourage
delegated representatives to acquire information and vote sincerely.
Both incentives matter for governance: compensating representatives
for their effort is useful only if the information they acquire
contributes to better collective decisions.

Conditioning rewards on a sufficiently large winning coalition
can make sincere voting strictly optimal while preserving the
effort supported by proportional sharing. Under suitable
reward-cost conditions, CPSR can also sustain higher effort
and more accurate collective decisions. These are equilibrium
existence results: insincere equilibria may coexist, and the
fixed-budget limitation on accuracy in large symmetric
committees remains.

Equilibrium selection is particularly important because greater
accuracy need not benefit the DReps themselves. At the symmetric
equilibria compared here, each DRep receives the same expected
gross payment, so higher effort reduces her expected net payoff:
\[
\frac{B}{2n+1}-ax^*
<
\frac{B}{2n+1}-ax^{psr}
\qquad\text{when }x^*>x^{psr}.
\]
CPSR also supports a sincere equilibrium at the lower effort
$x^{psr}$. Thus, every DRep prefers that equilibrium in expected
net-payoff terms, even when the higher-effort equilibrium yields
more accurate decisions. This comparison does not undermine
individual optimality within either equilibrium, but shows why
existence alone does not determine which outcome will emerge.
Appendix~\ref{app:implementation} outlines a possible implementation
path for Cardano, combining temporary ballot confidentiality with
public participation tracking and CPSR reward settlement.

Several research directions follow. First, repeated voting and changes
in delegation could help explain equilibrium selection and
whether reputational incentives favor informed, sincere behavior.
Second, supermajority requirements and penalties for supporting
the losing alternative could alter effort incentives, but their
effects must be assessed jointly with strict SVIC and participation
constraints. Finally, implementation deserves attention: when
interim votes are observable, representatives may wait and imitate
others rather than acquire information. Studying whether concealed
votes until a common deadline can preserve the incentives identified
here would connect the simultaneous-voting model to practical
governance procedures.

%% file: sections/appendix-proof-short.tex
\section{Missing Proofs}\label{appendix:proofs}

%------------------------------------------------------
\subsection{Proof of Theorem \ref{pr:strict-impossibility}}
\begin{proof}

Fix $\mu\in(0,1)\setminus\{1/2\}$ and consider a positive effort profile satisfying the inequalities in the statement. Such a profile exists because $0<|\mu-1/2|<1/2$.

At this profile, all DReps other than $i$ sincerely reject after either signal if $\mu<1/2$, and sincerely accept after either signal if $\mu>1/2$. Holding these strategies fixed, each action of DRep $i$ determines the winning coalition and her transfer, apart from any independent tie-breaking coin. Since payments depend neither on her private signal nor on the true state, acceptance gives the same expected transfer $T_A$ after either signal, and rejection gives the same expected transfer $T_R$ after either signal. This follows from the payment structure, without imposing SVIC; it does
not yet imply that $T_A=T_R$.

DRep $i$, however, sincerely accepts after $g$ and rejects after $b$. SVIC therefore requires $T_A\geq T_R$ after $g$ and $T_R\geq T_A$ after $b$, implying
$T_A=T_R$. Thus, any rule satisfying SVIC at this profile must leave DRep $i$ indifferent between the two actions after either signal.

Both signals have positive probability and both posteriors differ from $1/2$, so this indifference rules out strict SVIC. Hence strict SVIC cannot hold at this profile and therefore cannot hold at every positive effort profile.
\end{proof}

%-------------------------------------------------------
\subsection{Proof of Theorem \ref{pr:accuracy}}
\begin{proof}
First, PC and BC give
\[
a\sum_i x_i^*
\leq\sum_i\mathbb E[f_i(W,\mathbf x^*)]
=\mathbb E\!\left[\sum_i f_i(W,\mathbf x^*)\right]
\leq B.
\]
Hence $x^*\leq B/[a(2n+1)]$ at a symmetric profile.

Next suppose $\mu\geq1/2$. By the definition of success,
\begin{align}
\Pr(\text{success};v^*,\mathbf x^*)
&=\mu+(1-\mu)\Pr(\text{reject}\mid B)
       -\mu\Pr(\text{reject}\mid G)\notag\\
&\leq\mu+(1-\mu)
       \bigl[\Pr(\text{reject}\mid B)
                  -\Pr(\text{reject}\mid G)\bigr].
\label{pr:prior-bound}
\end{align}
We bound the difference in brackets by considering the full vector of private signals. Hold the prior, effort profile, and voting strategies fixed. For each signal vector, the probability of rejection
is a number in $[0,1]$. It is the same whether that vector arises in state $B$ or state $G$, because votes cannot additionally depend on the unobserved true state. Only the probability of observing
the vector changes across states. The rejection-probability difference is therefore a weighted sum of these differences in signal-vector probabilities, with weights in $[0,1]$. Dropping negative terms and replacing the remaining weights by one gives an upper bound: the sum of the positive differences between the probabilities
under $B$ and under $G$.

Put $p=1/2+x^*$. When $0<x^*<1/2$, a signal vector with
$k$ bad signals has probabilities
$p^k(1-p)^{2n+1-k}$ under $B$ and
$p^{2n+1-k}(1-p)^k$ under $G$.
The first exceeds the second exactly when $k\geq n+1$. 
Summing these positive differences, let
\[
F_n(x):=\sum_{k=n+1}^{2n+1}\binom{2n+1}{k}
\left(\frac12+x\right)^k
\left(\frac12-x\right)^{2n+1-k}.
\]
The probability of a majority of bad signals is $F_n(x^*)$
in state $B$ and $1-F_n(x^*)$ in state $G$, by symmetry and the odd committee size. Hence,
\begin{align*}
\Pr(\text{reject}\mid B)-\Pr(\text{reject}\mid G)
&\leq F_n(x^*)-[1-F_n(x^*)],\\
&=2F_n(x^*)-1,\\
&=2\sum_{k=n+1}^{2n+1}
\binom{2n+1}{k}p^k(1-p)^{2n+1-k}-1.
\end{align*}

At $x^*=0$ the signal distributions are identical, so the
difference is zero. At $x^*=1/2$ their supports are disjoint,
and the bound is one. Thus the displayed inequality also
covers both endpoints.

To connect this bound to the effort budget, we use
 \(F_n(x^*)-F_n(0) \leq x^*\max_{t\in[0,x^*]}F_n'(t),\)
where $F_n(0)=1/2$. Differentiating gives
\[
\frac{d}{dx}
\sum_{k=n+1}^{2n+1}
\binom{2n+1}{k}
\left(\frac12+x\right)^k
\left(\frac12-x\right)^{2n+1-k}
=
(2n+1)\binom{2n}{n}\left(\frac14-x^2\right)^n.
\]
Indeed, use
$k\binom{2n+1}{k}=(2n+1)\binom{2n}{k-1}$
and
$(2n+1-k)\binom{2n+1}{k}=(2n+1)\binom{2n}{k}$;
all interior terms cancel.
The tail equals $1/2$ at $x=0$, so (using the fundamental theorem of calculus $F_n(x^*)-F_n(0)=\int_0^{x^*}F_n'(t)\,dt.$)
\begin{align}
&\sum_{k=n+1}^{2n+1}\binom{2n+1}{k}
 \left(\frac12+x^*\right)^k
 \left(\frac12-x^*\right)^{2n+1-k}-\frac12
\notag\\
&\qquad=(2n+1)\binom{2n}{n}
       \int_0^{x^*}\left(\frac14-x^2\right)^n\,dx
\leq (2n+1)x^*\frac{\binom{2n}{n}}{4^n}.
\label{pr:tail-integral}
\end{align}
Combining this inequality with \eqref{pr:prior-bound} and the
effort budget proves \eqref{pr:accuracy-bound} for $\mu\geq1/2$.
Exchanging the states and the two decisions proves the other case.
The asymptotic claim follows from
$\binom{2n}{n}/4^n\sim1/\sqrt{\pi n}$.
\end{proof}

%------------------------------------------------------

% \begingroup
% \setlength{\emergencystretch}{1em}
\subsection{Proof of Proposition \ref{pr:PSR-properties}}
\begin{proof}
Rewards sum to $B$ whenever total effort is positive.
Since neither a DRep's vote nor the other votes affect her reward,
both voting actions are optimal at every information set.
Sincere continuations are therefore rational after every effort
profile, but sincerity is not strictly optimal.

Against effort $x>0$, a DRep's payoff is
\[
B\frac{x_i}{x_i+2nx}-ax_i.
\]
Its first derivative is $2nBx/(x_i+2nx)^2-a$ and its second
derivative is $-4nBx/(x_i+2nx)^3<0$.
At $x_i=x$, the first derivative is
$2nB/[(2n+1)^2x]-a$.
Strict concavity gives the stated symmetric global best response:
the derivative is zero at the interior value, and nonnegative
at $x=1/2$ precisely when that value is at least the effort cap.
This proves ER. Finally,
\[
\pi_i=\frac{B}{2n+1}-ax^{psr}
\geq\frac{B}{(2n+1)^2}>0,
\]
which proves PC.
\end{proof}

%------------------------------------------------------

\subsection{Proof of Proposition~\ref{pr:CPSR-properties}}
\label{app:cpsr-basic}
\begin{proof}
An eligible winning coalition shares exactly $B$; otherwise all
rewards are zero. This establishes BC.

Suppose all DReps other than $i$ reject independently of their
signals. Joining rejection gives $Bx_i/\sum_jx_j$ when total effort
is positive, and zero otherwise. Accepting gives zero: either the
accepting coalition loses or it is a singleton and fails the cardinality
requirement. The conclusion also holds at a weighted tie, whichever
alternative the fair coin selects. Pooling on rejection is therefore
a voting-stage equilibrium at every profile. The same reasoning applies
to acceptance, including zero-effort profiles. If $x_i>0$, joining
the pooling vote gives a strictly positive payment, while dissenting
gives zero, so joining is the unique best response. This is strict
SVIC only when that pooling action is sincere.

At a symmetric profile with $|\mu-1/2|<x\leq1/2$, sincere votes
follow signals. Put $p=1/2+x$ and let $K$ count other DReps voting to accept.
Conditional on $G$, $K\sim\operatorname{Binomial}(2n,p)$; in state
$B$, replace $p$ by $1-p$. The payment difference between acceptance
and rejection, divided by $B$, at $K=k$ is
\[
h_k=\frac{\mathbf{1}\{k\geq n\}}{k+1}
-\frac{\mathbf{1}\{k\leq n\}}{2n-k+1}.
\]
Pairing $k=n+r$ with $k=n-r$ gives
\[
H(p)=\sum_{r=1}^{n}\frac{\binom{2n}{n+r}}{n+r+1}
\left[p^{n+r}(1-p)^{n-r}-p^{n-r}(1-p)^{n+r}\right]>0.
\]
Every term is positive for $1/2<p<1$, and $H(1)=1/(2n+1)>0$.
Also $H(1-p)=-H(p)$. At posterior $\hat\mu_{i,s}$,
\[
U_i^A(s)-U_i^R(s)=B(2\hat\mu_{i,s}-1)H(p).
\]
The sign agrees with the posterior preference, establishing strict SVIC.

It remains to cover symmetric profiles at priors outside the strict
signal-contingent region. If $\mu\leq1/2-x$, sincere voting pools
on rejection; if $\mu>1/2+x$, it pools on acceptance. In either
case, joining the sincere vote gives $B/(2n+1)>0$, while dissenting
gives zero. Thus strict SVIC holds at these profiles as well.
At $\mu=1/2+x$, the convention of rejecting at posterior $1/2$
makes sincere voting follow signals. The same expression involving
$H(p)>0$ applies: acceptance is uniquely optimal after a good signal,
and both votes are optimal after a bad signal. This satisfies strict
SVIC as defined, since uniqueness is required only when the posterior
differs from $1/2$. The boundary priors belong to $(0,1)$ only
when $x<1/2$. Together, these cases establish strict SVIC for every
symmetric positive profile and every interior prior.

The failure of ordinary SVIC at an asymmetric positive profile is
established in Example~\ref{pr:SVIC-example} (see Appendix~\ref{app:cpsr-asymmetric}).

Finally, prescribe the same pooling continuation after every effort
profile and put $x_0=x^{psr}$. When the other DReps exert effort $x_0$,
the best continuation vote is to join them, giving
\[
P(z;x_0)=B\frac{z}{z+2nx_0}-az,
\qquad
\frac{\partial P}{\partial z}=\frac{2nBx_0}{(z+2nx_0)^2}-a.
\]
This payoff is strictly concave. Its derivative is zero at $z=x_0$
when $x_0<1/2$, and nonnegative there when $x_0=1/2$. Thus $x_0$
is a global best response on $[0,1/2]$. Voting deviations cannot
improve the payoff, so ER holds. Moreover,
\[
\frac{B}{2n+1}-ax_0\geq\frac{B}{(2n+1)^2}>0,
\]
which gives PC. Applying either pooling continuation yields the two
claimed equilibria. Sincere voting pools on rejection if
$\mu\leq1/2-x_0$, and on acceptance if $\mu>1/2+x_0$.
The strictly positive gain from joining these sincere pooling votes
establishes strict SVIC in the respective prior regions.
\end{proof}

%\subsection{Why the cardinality requirement matters for effort incentives}
%\label{app:cpsr-capture}
%Let $n=1$, $B=a=1$, and $x^{psr}=2/9$. Choose
%$0<\eta<1/18$ and $0<\mu<1/18-\eta$. At the symmetric profile,
%sincere voting pools on rejection and each DRep obtains $1/9$.
%If the cardinality condition is removed, a deviation to
%$z=4/9+\eta$ gives one DRep more than half the voting weight.
%By accepting alone she receives the entire budget and obtains
%\[
%1-z=\frac59-\eta>\frac19.
%\]
%The prior condition ensures that even after a good signal her sincere
%action remains rejection; the profitable vote is insincere.
%Restoring the cardinality requirement gives this singleton zero
%payment. Under pooling on rejection after every effort profile,
%joining the others yields $z/(z+4/9)-z$, maximized at $z=2/9$.
%Thus no unilateral effort or voting deviation is profitable under
%that continuation prescription.

%------------------------------------------------------

\subsection{Proof of Theorem~\ref{thm:cpsr-benchmark-sincere}}
\begin{proof}
Write $x_0=x^{psr}$. Following the effort profile
$\bar{\mathbf{x}}=(x_0,\ldots,x_0)$, prescribe sincere voting.
Following every other effort profile, prescribe pooling on rejection. A positive-effort DRep can infer the profile from her effort and the
public weights; a zero-effort DRep knows that the symmetric positive profile was not reached and selects pooling.

At $\bar{\mathbf{x}}$, strict SVIC holds for every interior prior by Proposition~\ref{pr:CPSR-properties}, including the threshold values separating sincere pooling from signal-contingent voting. Hence sincere voting is a voting-stage equilibrium there, with a unique best response
whenever the posterior differs from $1/2$. At every other effort profile, pooling on rejection is a voting-stage equilibrium by the same proposition. Thus the prescribed voting continuations are sequentially rational, although they need not be sincere after effort deviations.

At $\bar{\mathbf{x}}$, equal voting weights imply that the winning coalition contains at least $n+1$ DReps, so the entire budget $B$ is distributed in every realization. By symmetry, each DRep's expected
gross payment is $B/(2n+1)$, whether sincere voting is pooling or signal-contingent. Thus her equilibrium net payoff is
\[
\pi_i^*=\frac{B}{2n+1}-ax_0.
\]

Consider any unilateral effort deviation to $z\neq x_0$. The other
DReps reject under the prescribed continuation. Joining them gives
DRep $i$ the payoff
\[
P(z;x_0)=B\frac{z}{z+2nx_0}-az.
\]
Voting to accept instead gives zero gross payment: either rejection
wins, or DRep $i$ forms a singleton winning coalition that fails the
cardinality requirement. Consequently, $P(z;x_0)$ is her maximum
payoff over all continuation voting rules, including randomized rules.

The function $P(z;x_0)$ is strictly concave in $z$, with
\[
\frac{\partial P(z;x_0)}{\partial z}
=\frac{2nBx_0}{(z+2nx_0)^2}-a.
\]
If $x_0<1/2$, this derivative is zero at $z=x_0$. If $x_0=1/2$,
it is nonnegative at the upper bound, so the constrained maximum is
again attained at $z=x_0$. Therefore,
\[
P(z;x_0)\leq P(x_0;x_0)
=\frac{B}{2n+1}-ax_0=\pi_i^*
\qquad\text{for every }z\in[0,1/2].
\]
A deviation in voting alone is unprofitable by on-path voting
optimality. Hence, no joint effort and voting deviation is profitable,
which establishes ER.

Under the conditional proportional sharing rule, total payments equal
$B$ when the winning coalition is payment-eligible and zero otherwise,
so BC holds. Finally, the definition of $x_0$ implies
$ax_0\leq 2nB/(2n+1)^2$, and therefore
\[
\pi_i^*=\frac{B}{2n+1}-ax_0
\geq\frac{B}{(2n+1)^2}>0,
\]
which establishes PC. Since the effort choices and continuation rules
are invariant under relabeling DReps, the equilibrium is symmetric.
\end{proof}

%------------------------------------------------------

\subsection{Proof of Lemma~\ref{lem:sincere_voting_effort_deviation}}
\begin{proof}
\textit{Part 1.}
Suppose all DReps $j\neq i$ vote sincerely. Since
$|\mu-1/2|<x^*$, they accept after a good signal and reject after a
bad signal. Put $p^*=1/2+x^*$ and $q^*=1/2-x^*$. If $K$ denotes
the number of these DReps who vote to accept, then
\[
K\mid G\sim\mathrm{Binomial}(2n,p^*),
\qquad
K\mid B\sim\mathrm{Binomial}(2n,q^*).
\]

Consider first $x_i>0$. If DRep $i$ accepts, her coalition is
payment-eligible and wins whenever $K=k\geq n$, since
$x_i+kx^*-(2n-k)x^*=x_i+2(k-n)x^*>0$.
For $k<n$, her coalition receives no payment even if it wins by weight.
If she rejects, the corresponding condition is $k\leq n$.
Thus, define
\[
h_k(x_i,x^*):=
\frac{\mathbf{1}\{k\geq n\}}{x_i+kx^*}
-\frac{\mathbf{1}\{k\leq n\}}{x_i+(2n-k)x^*},
\]
so that $Bx_i h_k(x_i,x^*)$ is her payment advantage from accepting
rather than rejecting. Conditional on state $G$, write
\[
H(x_i,x^*;p^*)=
\sum_{k=0}^{2n}\binom{2n}{k}(p^*)^k(q^*)^{2n-k}h_k(x_i,x^*).
\]
Pairing $k=n+r$ with $k=n-r$ gives
\[
H(x_i,x^*;p^*)=
\sum_{r=1}^{n}\binom{2n}{n+r}
\frac{(p^*)^{n+r}(q^*)^{n-r}
      -(p^*)^{n-r}(q^*)^{n+r}}
     {x_i+(n+r)x^*}>0.
\]
The terms are nonnegative and at least one is strictly positive,
including when $p^*=1$. Reversing the states changes the sign of this
expression. Hence, at posterior $\hat\mu_i$, her expected payment
advantage from accepting rather than rejecting is
\[
Bx_i(2\hat\mu_i-1)H(x_i,x^*;p^*).
\]
Since $x_i>0$ and $H(x_i,x^*;p^*)>0$, posterior-sincere voting
is the unique best response whenever $\hat\mu_i\neq1/2$.
At $\hat\mu_i=1/2$, the payment advantage is zero, so the prescribed
sincere rejection is weakly optimal. At $x_i=0$, both actions yield
zero payment and only weak optimality can be asserted.

\medskip
\noindent\textit{Part 2.}
Write $x=x^*$ and $z=x_i$, and suppose $|\mu-1/2|<z\leq1/2$.
All DReps' sincere voting rules follow their signals. Part~1 establishes
DRep $i$'s voting incentives. Consider any DRep $j\neq i$, whose effort
is $x$. First let $x<1/2$, and put
\[
p=\frac12+x,\quad q=\frac12-x,\quad
r=\frac12+z,\quad s=\frac12-z,\quad \rho=\frac{z}{x}>0.
\]
Among the DReps other than \(j\), let $K$ count the accepting DReps with effort
$x$, and let $L$ indicate acceptance by DRep $i$. Conditional on $G$,
these variables are independent, with
\[
K\sim\mathrm{Binomial}(2n-1,p),\qquad L\sim\mathrm{Bernoulli}(r).
\]
If $j$ accepts and $L=1$, her coalition contains $k+2$ DReps.
It meets the cardinality requirement if and only if $k\geq n-1$;
then its weight exceeds the opposing weight by $z+2(k-n+1)x>0$.
Her payment is consequently $B/(k+1+\rho)$.
If $j$ accepts and $L=0$, cardinality requires $k\geq n$.
A strict victory requires $2(k-n+1)x>z$, and equality is
resolved by the fair coin. Thus, define
\[
\lambda_k(\rho):=
\begin{cases}
1,&\rho<2(k-n+1),\\
1/2,&\rho=2(k-n+1),\\
0,&\rho>2(k-n+1),
\end{cases}
\qquad k=n,\ldots,2n-1.
\]
Her expected payment in this case is $B\lambda_k(\rho)/(k+1)$.
Both victory by weight and the cardinality requirement are imposed.

For $k=0,\ldots,2n-1$, write
\[
\alpha_k=\binom{2n-1}{k}p^kq^{2n-1-k},\qquad
\beta_k=\binom{2n-1}{k}q^kp^{2n-1-k}.
\]
Reversing the votes of all DReps other than $j$ in the
payment from rejection gives DRep $j$'s conditional expected
payment advantage from acceptance in state $G$, divided
by $B$:
\[
H_j=\sum_{k=n-1}^{2n-1}\frac{\alpha_k r-\beta_k s}{k+1+\rho}
    +\sum_{k=n}^{2n-1}\lambda_k(\rho)\frac{\alpha_k s-\beta_k r}{k+1}.
\]
We show that $H_j>0$ in both deviation directions.

\medskip
\noindent\emph{Downward deviations: $0<\rho\leq1$.}
Here $\lambda_k(\rho)=1$ for every $k=n,\ldots,2n-1$. Since $r/s\leq p/q$,
all numerators are nonnegative except possibly the first-sum term
$k=n-1$. Combine it with the $k=n$ terms in both sums, and put
$C_n=\binom{2n-1}{n}=\binom{2n-1}{n-1}$. Their sum is
\begin{align*}
&C_n(pq)^{n-1}
\left[\frac{x-z}{n+1}-\frac{x-z}{n+\rho}
                         +\frac{x+z}{n+1+\rho}\right]\\
&\qquad=C_n(pq)^{n-1}x
\left[\frac{1+\rho}{n+1+\rho}
      -\frac{(1-\rho)^2}{(n+1)(n+\rho)}\right]>0.
\end{align*}
The inequality uses $(n+1)(n+\rho)\geq n+1+\rho$ and
$1+\rho>(1-\rho)^2$. The remaining terms are nonnegative, so $H_j>0$.

\medskip
\noindent\emph{Upward deviations: $\rho>1$.}
Every first-sum numerator $\alpha_k r-\beta_k s$ is nonnegative
for $k\geq n-1$, and strictly positive for $k\geq n$.
For each $k$ with $\lambda_k(\rho)>0$, pair its second-sum term with the first-sum
term having the same index. Put
\[
d=2k-2n+1\geq1,\qquad
u_k=\frac{\alpha_k-\beta_k}{\alpha_k+\beta_k}.
\]
Since $\alpha_k/\beta_k=(p/q)^d$, Bernoulli's inequality gives
\[
\left(\frac{p}{q}\right)^d
=\left(1+\frac{4x}{1-2x}\right)^d
\geq1+\frac{4dx}{1-2x},\qquad
u_k\geq\frac{2dx}{1+2(d-1)x}.
\]
The condition $\lambda_k(\rho)>0$ implies $z\leq(d+1)x$. Since $z\leq1/2$,
\[
2dx-z[1+2(d-1)x]
=2x[d-(d-1)z]-z
\geq(d+1)x-z\geq0.
\]
Consequently, $u_k\geq z$. Giving the second term full weight yields
\begin{align*}
&\frac{\alpha_k s-\beta_k r}{k+1}
 +\frac{\alpha_k r-\beta_k s}{k+1+\rho}\\
&\quad=\frac{\alpha_k+\beta_k}{(k+1)(k+1+\rho)}
 \left[\left(k+1+\frac\rho2\right)u_k-\rho z\right]>0,
\end{align*}
because $k+1>\rho/2$ and the expression in brackets is at least
$z(k+1-\rho/2)>0$. At a tie, the pair has the form
$A+\tfrac12 C=\tfrac12 A+\tfrac12(A+C)>0$, where $A$ is its
strictly positive first-sum term and $A+C$ is the full pair just
bounded. Thus the half-weight term required by fair tie-breaking
preserves positivity. All unpaired first-sum terms are nonnegative.
If every $\lambda_k$ is zero, the first sum itself contains strictly
positive terms. Hence $H_j>0$ also after upward deviations.

\medskip
Reversing the states changes the conditional payment advantage to
$-BH_j$. Conditional independence of signals therefore implies that,
at posterior $\hat\mu_j$, the expected payment advantage from acceptance is
\[
B(2\hat\mu_j-1)H_j,
\]
Since $x,z>|\mu-1/2|$, every posterior generated by a signal of
positive probability differs from $1/2$. The sincere vote is therefore
the unique best response.

Finally, consider $x=1/2$ directly. DRep $j$ and the other
$2n-1$ DReps with effort $x$ observe the state perfectly. Put
$\rho=2z\in(0,1]$, $r=1/2+z$, and $s=1/2-z$.
Conditional on state $G$, if $j$ accepts, she receives
$B/(2n+\rho)$ when $i$ accepts and $B/(2n)$ when $i$ rejects.
Both coalitions win by weight and contain at least $n+1$ members.
If $j$ rejects, she can be paid only when $n=1$ and $i$ also
rejects: that two-member coalition wins by weight and gives her
$B/(1+\rho)$. For $n\geq2$, her rejecting coalition has at most
two members and fails the cardinality requirement. Consequently,
\[
H_j=\frac{r}{2n+\rho}+\frac{s}{2n}
-\mathbf{1}\{n=1\}\frac{s}{1+\rho}.
\]
For $n\geq2$, the first term is strictly positive and the second
is nonnegative. For $n=1$, direct simplification gives
\[
H_j=\frac{\rho(7+2\rho-\rho^2)}{4(1+\rho)(2+\rho)}>0
\qquad (0<\rho\leq1).
\]
State $B$ gives the opposite payment difference. Since $j$'s
posterior is either zero or one, sincerity is uniquely optimal.

The argument applies to every $j\neq i$. Together with Part~1,
it establishes strict SVIC and the claimed sincere voting-stage
equilibrium, including the maximal-effort endpoint.
\end{proof}

%------------------------------------------------------

\subsection{Proof of Proposition~\ref{prop:cpsr-interior-existence}}
\begin{proof}
Put $N=2n+1$, $R=B/a$, and $\delta=|\mu-1/2|$.
For common effort $x>0$ among the other $2n$ DReps,
write $p=1/2+x$ and $q=1/2-x$. When every DRep
votes according to her signal, the deviator's expected
gross payment divided by $B$ is
\[
g_n(z,x)
=
\sum_{k=n}^{2n}\binom{2n}{k}
\left[
(\tfrac12+z)p^kq^{2n-k}
+(\tfrac12-z)p^{2n-k}q^k
\right]\frac{z}{z+kx}.
\]
The restriction $k\geq n$ imposes the minimum coalition
size. For $z>0$, it also guarantees victory by weight,
since
\[
z+kx-(2n-k)x=z+2(k-n)x>0.
\]
If $k<n$, the deviator receives no payment even if her
coalition wins by weight. At $z=0$, her payment is zero.
The corresponding net payoff divided by $a$ is
\[
u_R(z,x):=R g_n(z,x)-z.
\]

\paragraph{Effort incentives near the cap.}
At $x=1/2$, the other DReps' signals reveal the state
perfectly. The deviator joins their coalition with
probability $1/2+z$, obtaining payment share $z/(z+n)$;
otherwise she receives zero. Hence
\[
g_n(z,1/2)=\frac{z(1/2+z)}{z+n},
\]
and differentiation gives
\[
\left.
\frac{\partial g_n(z,x)}{\partial z}
\right|_{z=x=1/2}
=
\frac1N+\frac{4n}{N^2}.
\]
Consequently,
\[
\left.
\frac{\partial u_R(z,x)}{\partial z}
\right|_{R=R_n,\ z=x=1/2}
=
\frac{N^2}{4n}
\left(\frac1N+\frac{4n}{N^2}\right)-1
=
\frac{N}{4n}>0.
\]
This derivative is continuous in $(R,z,x)$ in a feasible
neighborhood of the displayed point, since $g_n$ is a
finite sum with strictly positive denominators there.

For $R<R_n$, let
\[
x_0(R):=\frac{2nR}{N^2}
\]
denote the interior PSR effort. As $R\uparrow R_n$,
$x_0(R)\uparrow1/2$. Since $\delta<1/2$, continuity
therefore gives an $\varepsilon_n(\mu)\in(0,R_n)$
such that, for every
$R\in(R_n-\varepsilon_n(\mu),R_n)$,
\[
\delta<x_0(R)<\frac12,
\qquad
\left.
\frac{\partial u_R(z,x)}{\partial z}
\right|_{z=x=x_0(R)}>0.
\]

Fix such an $R$ and write $x_0=x_0(R)$.
For $x>x_0$, define
\[
\tau(x):=Nx_0-2nx.
\]
As $x\downarrow x_0$, both $x$ and $\tau(x)$
converge to $x_0$, and $\tau(x)<x$ whenever $x>x_0$.
The derivative remains positive near $(z,x)=(x_0,x_0)$.
We may therefore choose $x>x_0$ sufficiently close
to $x_0$ that
\[
\delta<\tau(x)<x<\frac12,
\qquad
\frac{\partial u_R(z,x)}{\partial z}>0
\quad\text{for every }z\in[\tau(x),x].
\]
Write $\tau=\tau(x)$ and take $x^*=x$.
The preceding inequality implies
\[
u_R(z,x)\leq u_R(x,x)
\qquad\text{for every }z\in[\tau,x].
\]

\paragraph{Voting continuations.}
Prescribe sincere voting at $(x,\ldots,x)$ and at
every permutation of $(z,x,\ldots,x)$ with $\tau<z<x$.
At every other effort profile, prescribe pooling
on rejection.

Since $\delta<\tau$, all DReps' sincere votes follow
signals at the prescribed sincere profiles.
Proposition~\ref{pr:CPSR-properties} and
Lemma~\ref{lem:sincere_voting_effort_deviation}
establish strict SVIC at those profiles, so these
continuations are voting-stage equilibria.
Pooling on rejection is a voting-stage equilibrium
at every other profile by
Proposition~\ref{pr:CPSR-properties}.
Thus all prescribed continuations are sequentially
rational, although the pooling continuations need
not be sincere.

The prescription is feasible given the available
information. A positive-effort DRep can infer the
effort profile from her own effort and the public
voting weights. A zero-effort DRep knows that none
of the prescribed sincere profiles was reached
and selects rejection.

\paragraph{Global effort incentives.}
At the symmetric profile, equal positive weights
ensure that the winning coalition contains at least
$n+1$ DReps. The entire budget is therefore paid.
By symmetry, each DRep's expected gross payment is
$B/N$, giving net payoff
\[
U:=\frac{B}{N}-ax.
\]

Following a deviation with $\tau<z<x$, sincere voting
is a best response under the prescribed continuation.
The deviator's maximum payoff over voting rules
therefore satisfies
\[
a u_R(z,x)\leq a u_R(x,x)=U.
\]

Following any other effort deviation, the other DReps
pool on rejection. Joining them gives net payoff
\[
P(z;x):=B\frac{z}{z+2nx}-az.
\]
Acceptance gives zero gross payment: the deviator
either loses by weight or forms a singleton winning
coalition that fails the minimum-size requirement.
Thus $P(z;x)$ is her maximum payoff over voting rules.

Since $x_0=2nB/(aN^2)$ and $\tau=Nx_0-2nx$,
\begin{align*}
P(z;x)-U
&=(z-x)\left[\frac{2nB}{N(z+2nx)}-a\right]\\
&=a\,\frac{(z-x)(\tau-z)}{z+2nx}
\leq0
\end{align*}
whenever $z\leq\tau$ or $z>x$.
Together with the sincere deviation region, these
cases exhaust all unilateral effort deviations.
A voting-only deviation is unprofitable by strict
SVIC at the symmetric profile. Hence no joint effort
and voting deviation is profitable, establishing ER.

\paragraph{Budget, participation, and symmetry.}
CPSR distributes $B$ when the weighted winning
coalition has at least $n+1$ members, and zero
otherwise, so BC holds. Moreover, $\tau>0$ implies
\[
x<\frac{Nx_0}{2n}=\frac{B}{aN},
\]
so $U>0$ and PC holds.

The effort choices and continuation prescription
are invariant under relabeling DReps, so the
equilibrium is symmetric. Finally,
\[
\delta<x_0=x^{psr}<x=x^*<\frac12,
\]
which proves the stated effort comparison and
ensures sincere, signal-contingent voting on
the equilibrium path.
\end{proof}

%------------------------------------------------------

\subsection{Proof of Corollary~\ref{cor:cpsr-higher-success}}
\label{app:cpsr-success}
\begin{proof}
Proposition~\ref{prop:cpsr-interior-existence} establishes
$|\mu-1/2|<x^{psr}<x^*<1/2$. Thus sincere votes follow signals
at both symmetric effort profiles. Equal positive weights implement
simple majority, and conditional on either state success probability is
\[
F_n(x)=\sum_{k=n+1}^{2n+1}\binom{2n+1}{k}
\left(\frac12+x\right)^k
\left(\frac12-x\right)^{2n+1-k}.
\]
Unconditional success probability is therefore $F_n(x)$ for every
prior in this SC region. Differentiating the
binomial tail and cancelling adjacent terms gives
\[
F_n'(x)=(2n+1)\binom{2n}{n}
\left(\frac14-x^2\right)^n>0
\qquad (0<x<1/2).
\]
Hence $F_n(x^*)>F_n(x^{psr})$. The conclusion applies to the
constructed sincere CPSR equilibrium and the sincere PSR equilibrium
at the same parameters.
\end{proof}

%% file: sections/appendix-results-examples-short.tex
\section{Additional Results and Examples}
\label{appendix:results-examples}
%-------------------------------------------------------
%------------------------------------------------------
\subsection{Additional analytical results}\label{app:additional-analytical-results}

%------------------------------------------------------
\subsubsection{Asymptotic accuracy under PSR}
\label{app:psr-accuracy}

\begin{corollary}[PSR asymptotically attains the accuracy bound]
\label{cor:psr-asymptotic-accuracy}
Let $\mu=1/2$ and fix $B,a>0$. At the sincere PSR equilibrium,
let
\[
x^{psr}
=
\min\left\{
\frac12,\frac{2nB}{a(2n+1)^2}
\right\},
\qquad
\boldsymbol{x}^{psr}
=
(x^{psr},\ldots,x^{psr}).
\]
Then
\[
\lim_{n\to\infty}
\frac{
\Pr(\mathrm{success};v^{sv},\boldsymbol{x}^{psr})-\frac12
}{
(B/a)\binom{2n}{n}/4^n
}
=1.
\]
\end{corollary}

%-----------------------------------------------------
\begin{proof}
Sincere votes follow the signals and equal weights implement
simple majority. Success therefore equals the binomial tail
in \eqref{pr:tail-integral}.
For all sufficiently large $n$, the effort cap does not bind and
$x^{psr}=2nB/[a(2n+1)^2]$.
On the interval $[0,x^{psr}]$, the integrand in
\eqref{pr:tail-integral} is between
$4^{-n}(1-4(x^{psr})^2)^n$ and $4^{-n}$.
Consequently the ratio in the statement lies between
\[
\frac{2n}{2n+1}\bigl(1-4(x^{psr})^2\bigr)^n
\quad\text{and}\quad
\frac{2n}{2n+1}.
\]
Since $n(x^{psr})^2\to0$, Bernoulli's inequality gives
$1\geq(1-4(x^{psr})^2)^n\geq1-4n(x^{psr})^2\to1$.
Both bounds converge to one.
\end{proof}
%-----------------------------

The numerator is the accuracy advantage of sincere PSR over
random choice. The denominator is the upper bound on that
advantage for symmetric equilibria under any reward rule
satisfying the budget and participation constraints. Their
ratio converges to one. Thus, with the same budget and effort
costs, any improvement on sincere PSR's accuracy becomes small
compared with PSR's own advantage over random choice.

Since
\[
\frac{\binom{2n}{n}}{4^n}
\sim \frac{1}{\sqrt{\pi n}},
\]
the corollary also gives
\[
\Pr(\mathrm{success};v^{sv},\boldsymbol{x}^{psr})
=
\frac12+\frac{B/a}{\sqrt{\pi n}}+o(n^{-1/2}).
\]
Hence sincere PSR approaches the accuracy bound while its
success probability converges to $1/2$. This result does not
rule out higher effort and accuracy under CPSR in finite
committees. It is specific to the balanced prior and does
not establish the same comparison with the bound when
$\mu\neq1/2$.

%------------------------------------------------------
\subsubsection{Maximal effort}
\label{app:maximal-effort}

Once $B/a$ reaches the PSR cap threshold, PSR already
supports maximal effort. CPSR can support the same effort
together with sincere voting for every interior prior.

\begin{proposition}[A symmetric sincere equilibrium with maximal effort]
\label{prop:maximal-effort}
Suppose $n\geq1$, $B>0$, and $a>0$. If
\[
\frac{B}{a}\geq\frac{(2n+1)^2}{4n},
\]
then, for every $\mu\in(0,1)$, CPSR admits a symmetric
equilibrium satisfying ER, strict SVIC, BC, and PC,
in which every DRep chooses $x^{**}=x^{psr}=1/2$
and votes sincerely on the equilibrium path.
All voting continuations are sequentially rational.
\end{proposition}

At maximal effort, signals reveal the state perfectly
and sincere voting is signal-contingent for every interior
prior. Every DRep accepts in state $G$ and rejects in
state $B$, so success probability is one. CPSR satisfies
strict SVIC at this equilibrium profile: following the
others' sincere votes gives payment $B/(2n+1)$, whereas
dissenting gives zero. Under PSR, sincere voting at the
same effort is optimal through indifference. Thus both
rules support maximal effort, while their voting
incentives differ.

\begin{proof}
At the stated reward-cost ratios, $x^{psr}=1/2$.
Theorem~\ref{thm:cpsr-benchmark-sincere} therefore supplies
the claimed equilibrium for every interior prior.
It satisfies ER, strict SVIC, BC, and PC, and all voting
continuations are sequentially rational. At this effort
every signal reveals the state, so all DReps vote correctly.
Following their sincere votes gives payment $B/(2n+1)>0$,
whereas dissenting gives zero. Thus strict SVIC holds
at the equilibrium profile and success probability is one.
\end{proof}

%------------------------------------------------------
%------------------------------------------------------

\subsection{Examples}\label{app:examples}

%-----------------------------------------------------
\subsubsection{The role of the cardinality requirement}

Voting weight determines which alternative wins, whereas
the cardinality requirement determines whether the winning
coalition receives a payment. The following example shows
why these conditions are distinct.

\begin{example}[The role of the cardinality requirement]
\label{ex:cardinality-requirement}
Let $n=1$, so the committee contains three DReps. Suppose
DReps 2 and 3 each exert effort $2/9$ and vote to reject.
If DRep 1 chooses
\[
x_1=\frac49+\varepsilon,
\qquad
0<\varepsilon<\frac1{18},
\]
her effort remains below $1/2$ and exceeds the combined
effort of the other two DReps. Her voting weight is therefore
\[
w_1
=
\frac{\frac49+\varepsilon}
     {\frac89+\varepsilon}
>
\frac12.
\]
By voting to accept, she makes acceptance win alone,
regardless of the other DReps' votes.

Without the cardinality requirement, proportional sharing
within the winning coalition would give her
\[
B\frac{x_1}{x_1}=B.
\]
Under CPSR, however, payment requires a winning coalition
with at least $n+1=2$ members. Since the accepting coalition
contains only DRep 1, she receives zero.
\end{example}

The cardinality requirement thus prevents a DRep from
capturing the entire budget through a singleton victory,
without restricting her ability to determine the voting
outcome. The example isolates this feature of the payment
rule; it does not assert that the effort choice is a
profitable deviation or an equilibrium.

%-------------------------------------------------------------
\subsubsection{Voting incentives across priors and effort profiles} \label{app:cpsr-asymmetric}
The following plot illustrates voting incentives at the
effort profile $x=x^{psr}$.

\input{figures/cpsr_prior_regions_n1}

Strict SVIC at symmetric positive profiles does not extend
to every asymmetric profile. The following example shows
that even ordinary SVIC can fail. We present the example first, followed by its supporting
calculations.

\begin{example}[Strict SVIC depends on the prior and effort profile]
\label{pr:SVIC-example}
Let $n=1$ and $\mathbf{x}=(0.3,0.2,0.2)$. At $\mu=1/4$,
DReps~2 and~3 sincerely reject after either signal. After
a good signal, DRep~1 has posterior $4/7>1/2$, so her
sincere action is acceptance. That action yields zero
payment, whereas joining rejection yields $3B/7$.
Thus CPSR fails SVIC at this profile. At the same effort
profile, strict SVIC holds for $0.3<\mu<0.7$, when all
sincere votes follow signals. These are voting-stage
comparisons; no effort-equilibrium claim is made for
this profile.
\end{example}

\paragraph{Supporting calculations.}
Fix $\mathbf{x}=(0.3,0.2,0.2)$. Every pair outweighs
the remaining singleton and contains the two members
required for payment. When $0.3<\mu<0.7$, all sincere
votes follow signals. In state $G$, the expected-payment
advantage of acceptance for DRep~1 is
\[
(0.7^2-0.3^2)\frac{3B}{7}=\frac{6B}{35}.
\]
When the other DReps disagree, she receives the same
share by joining either one, so those configurations
cancel. For DRep~2 or~3, the advantage is
\[
(0.8\cdot0.7-0.2\cdot0.3)\frac{2B}{7}
+(0.8\cdot0.3-0.2\cdot0.7)
\left(\frac{2B}{5}-\frac{B}{2}\right)
=\frac{93B}{700}.
\]
The first term covers agreement by the other DReps;
the second covers disagreement and the different shares
obtained by joining each one. Reversing the state negates
these advantages. Hence
\[
U_i^A(s)-U_i^R(s)
=
B(2\hat\mu_{i,s}-1)
\begin{cases}
6/35,&i=1,\\
93/700,&i=2,3,
\end{cases}
\]
which proves strict SVIC in the stated prior interval.

At $\mu=1/4$, the sincere votes of DReps~2 and~3 are
always rejection, whereas
\[
\hat\mu_{1,g}
=
\frac{(1/4)(0.8)}
     {(1/4)(0.8)+(3/4)(0.2)}
=
\frac47>\frac12.
\]
Sincere acceptance gives DRep~1 zero: her singleton
loses by weight and also fails the cardinality requirement.
Rejecting instead gives $3B/7>0$, proving failure of SVIC.

More generally, away from posterior ties, strict SVIC
holds for $0<\mu<0.2$ and $0.8<\mu<1$, where all sincere
votes pool, and for $0.3<\mu<0.7$, as proved above.
SVIC fails for $0.2<\mu<0.3$ and $0.7<\mu<0.8$:
DRep~1 sometimes sincerely opposes the fixed votes of
DReps~2 and~3 and benefits from joining them.
No claim about ER at this effort profile is required.

%---------------------------------------------------------
\subsubsection{Interior effort comparisons}
\label{app:interior-effort-comparisons}
The following example compares effort under CPSR and PSR conditional on sincere, signal-contingent voting. It shows that CPSR can support higher effort, but that the comparison can reverse as committee size increases. Supporting
calculations follow the example.

\begin{example}[Interior effort comparisons]
\label{ex:interior_SC_equilibria}
Consider $n=1$, $\mu=1/2$, and positive $B,a$ satisfying
\[
\frac{72}{67}<\frac{B}{a}<\frac{72}{59}.
\]
We examine whether conditioning rewards on the winning coalition
can support greater individual effort than PSR when every DRep
votes according to her signal.

There is a unique symmetric effort candidate
$x^*\in(1/4,1/3)$ at which no DRep can improve her payoff
by changing her effort, holding the other DReps' efforts
fixed and maintaining signal-following voting. Moreover,
\[
x^{psr}=\frac{2B}{9a}<x^*<\frac{B}{3a}.
\]
Thus, in this case, the conditional effort incentives support
higher effort than PSR.

For $B=1.2$, $a=1$, and $\mu=1/2$, the table and plot extend
the comparison across committee sizes. Each reported $x^*$
is individually optimal against the other DReps exerting
$x^*$, conditional on sincere, signal-contingent voting.

\begin{figure}[H]
\centering
\begin{minipage}[c]{.39\linewidth}
\centering\small
\setlength{\tabcolsep}{3pt}
\begin{tabular}{rrr}
\toprule
$n$ & $x^{psr}$ & $x^*$\\
\midrule
1    & .2667    & .3174\\
2    & .1920    & .2069\\
3    & .1469    & .1528\\
10   & .0544    & .0543\\
20   & .0286    & .0284\\
50   & .01176   & .01171\\
1000 & .0005994 & .0005991\\
\bottomrule
\end{tabular}
\end{minipage}\hfill
\begin{minipage}[c]{.59\linewidth}
\input{figures/example3_plot.tex}
\end{minipage}
\caption{Conditional effort comparison at $\mu=1/2$,
$B=1.2$, and $a=1$. The plot shows $100(x^*/x^{psr}-1)$.
Calculations cover $n=1,\ldots,50$ and $n=1000$; the dashed
connector does not represent additional computed observations.}
\label{fig:cpsr-example3-expanded}
\end{figure}
\end{example}

The effort comparison reflects the interaction between the
eligibility and share margins. Under sincere, signal-contingent
voting, greater signal accuracy increases the probability of
joining a payment-eligible winning coalition. However, sharing
rewards only within that coalition gives a smaller marginal
return from increasing one's payment share than PSR does at
the same symmetric effort. The eligibility incentive dominates
for the smaller committees in the example. For the larger
reported committees, it is too weak to offset the weaker share
incentive, and the computed effort falls slightly below PSR.
For very large committees with weakly informative signals,
a DRep is paid approximately half the time and shares the
reward with approximately half the committee. The marginal
payment incentive is then close to that under PSR, explaining
why the relative effort gap narrows again.

These effort comparisons are conditional on sincere,
signal-contingent voting. When $\mu\neq1/2$, a sufficiently
large effort reduction makes the deviator's sincere vote
signal-independent; Example~\ref{ex:SI_deviations} examines
these deviations while holding the other DReps' sincere
strategies fixed. Example~\ref{ex:cpsr-sincere-equilibria}
combines credible sincere and pooling continuations to
establish full equilibria over suitable prior ranges.

\paragraph{Supporting calculations: conditional payoff.}
We first give the payoff formula for an arbitrary committee
size $2n+1$, with $n\geq1$. Let DRep $i$ choose effort
$z\in[0,1/2]$, while the other $2n$ DReps choose common
effort $x\in(0,1/2]$. Put $p=1/2+x$ and $q=1/2-x$,
and suppose every DRep votes according to her signal.

Let $k$ count the other DReps casting the same vote as
DRep $i$. Payment requires $k\geq n$. For $z>0$, this
also guarantees victory by voting weight, since
\[
z+kx>(2n-k)x.
\]
Her payment is then $Bz/(z+kx)$. If $k<n$, a weighted
victory alone is insufficient for payment. Consequently,
\eqref{eq:conditional_expected_payoff} simplifies to
\begin{equation}
\label{eq:cpsr-sc-payoff}
\pi_i^{SC}(z,x)
=
B\sum_{k=n}^{2n}\binom{2n}{k}
\left[
(\tfrac12+z)p^kq^{2n-k}
+(\tfrac12-z)p^{2n-k}q^k
\right]\frac{z}{z+kx}-az.
\end{equation}
The two terms in brackets correspond to DRep $i$ receiving
a correct or an incorrect signal, respectively. Signal
accuracy and the payment rule are symmetric across states,
so this expected payoff is the same in either state and
does not depend on the prior. At $z=0$, her payment and
effort cost are both zero.

The formula evaluates signal-following behavior. When
$|\mu-1/2|<\min\{z,x\}$, that behavior is sincere and
Lemma~\ref{lem:sincere_voting_effort_deviation} establishes
a voting-stage equilibrium. Outside that region, the
formula remains a conditional payoff calculation, without
asserting that the prescribed votes are sincere or optimal.

\paragraph{The three-DRep calculation.}
\label{app:cpsr-n1}
Now let $n=1$, $R=B/a$, and $72/67<R<72/59$.

When the other two DReps' votes differ, either vote by
DRep $i$ joins a pair, giving payment $Bz/(z+x)$ for
$z>0$. This event has probability $2pq$ in either state.
All three votes agree with probability
$(1/2+z)p^2+(1/2-z)q^2$, giving payment $Bz/(z+2x)$.
Every other configuration gives the deviator zero payment.
These expressions also give zero payment at $z=0$.
Thus,
\[
\pi_i^{SC}(z,x)
=
B\left[
2pq\frac{z}{z+x}
+\left((\tfrac12+z)p^2+(\tfrac12-z)q^2\right)
\frac{z}{z+2x}
\right]-az.
\]
Substituting $p=1/2+x$ and $q=1/2-x$ and simplifying gives
\begin{equation}
\label{eq:cpsr-n1-payoff}
\pi_i^{SC}(z,x)
=
B\left[
2xz
+\left(\frac12-2x^2\right)\frac{z}{z+x}
+\left(\frac14-3x^2\right)\frac{z}{z+2x}
\right]-az.
\end{equation}

\paragraph{Symmetric effort candidate.}
Differentiating in $z$ and then setting $z=x$ gives
\[
\left.
\frac{\partial\pi_i^{SC}(z,x)}{\partial z}
\right|_{z=x}
=
B\left(\frac{5x}{6}+\frac{13}{72x}\right)-a.
\]
A symmetric interior candidate therefore satisfies
\begin{equation}
\label{eq:cpsr-n1-root}
B\left(\frac{5x^*}{6}+\frac{13}{72x^*}\right)=a.
\end{equation}
The bracket is strictly decreasing on $[1/4,1/3]$,
since its derivative is
\[
\frac56-\frac{13}{72x^2}<0.
\]
Its endpoint values are $67/72$ and $59/72$, so the
assumed interval for $R$ yields exactly one root
$x^*\in(1/4,1/3)$.

Multiplying \eqref{eq:cpsr-n1-root} by $x^*/B$ gives
\[
\frac{ax^*}{B}
=
\frac{13}{72}+\frac{5(x^*)^2}{6}
\in
\left(\frac{67}{288},\frac{59}{216}\right)
\subset
\left(\frac29,\frac13\right).
\]
Hence
\[
x^{psr}=\frac{2B}{9a}<x^*<\frac{B}{3a}.
\]

\paragraph{Global optimality under signal-following voting.}
The second derivative of \eqref{eq:cpsr-n1-payoff},
holding the other DReps' effort at $x^*$, is
\[
\frac{\partial^2\pi_i^{SC}(z,x^*)}{\partial z^2}
=
-2Bx^*\left[
\frac{1/2-2(x^*)^2}{(z+x^*)^3}
+\frac{1/2-6(x^*)^2}{(z+2x^*)^3}
\right].
\]
The first numerator in brackets is positive. If the second
is nonnegative, the bracket is positive. Otherwise,
replacing its denominator by $(z+x^*)^3$ gives a strict
lower bound on the bracket:
\[
\frac{1-8(x^*)^2}{(z+x^*)^3}>0,
\]
since $x^*<1/3$. Thus the payoff is strictly concave
throughout $[0,1/2]$. Together with
\eqref{eq:cpsr-n1-root}, this establishes that $x^*$
is the unique global maximizer:
\[
x^*\in\arg\max_{z\in[0,1/2]}\pi_i^{SC}(z,x^*).
\]
This verifies effort optimality conditional on
signal-following voting.

\paragraph{Numerical support for $n=3$.}
\label{app:cpsr-example3-n3}
We use $n=3$, $B=6/5$, and $a=1$ to illustrate the numerical
procedure for the other committee sizes. Calculations use
unrounded values; displayed numerical approximations are rounded.
Let the six other DReps choose effort $x>0$, and write
$p=1/2+x$ and $q=1/2-x$. The deviator chooses $z\in[0,1/2]$.

Conditional on sincere, signal-contingent voting, her
probability of voting correctly is $1/2+z$. If $k$ other
DReps cast the same vote, her coalition has $k+1$ members.
Payment requires $k\geq3$. For $z>0$, this also guarantees
victory by weight, since
\[
z+kx-(6-k)x=z+2(k-3)x>0.
\]
If $k<3$, winning by weight is insufficient for payment.
At $z=0$, the deviator receives zero payment.

Define
\[
A_k:=\binom6k p^kq^{6-k},\qquad
C_k:=\binom6k p^{6-k}q^k,\qquad k=3,4,5,6.
\]
These are the probabilities of $k$ other DReps sharing
the deviator's vote conditional on her vote being correct
or incorrect, respectively. Thus the four payment-event
probabilities are
\begin{align*}
w_3(z)&=20p^3q^3,\\
w_4(z)&=15\left[
  \left(\tfrac12+z\right)p^4q^2
  +\left(\tfrac12-z\right)p^2q^4
\right],\\
w_5(z)&=6\left[
  \left(\tfrac12+z\right)p^5q
  +\left(\tfrac12-z\right)pq^5
\right],\\
w_6(z)&=\left(\tfrac12+z\right)p^6
       +\left(\tfrac12-z\right)q^6.
\end{align*}
The conditional expected net payoff is therefore
\begin{equation}
\pi_i^{SC}(z,x)
=B\sum_{k=3}^{6}w_k(z)\frac{z}{z+kx}-az.
\label{eq:cpsr-n3-sc-payoff}
\end{equation}
The expression is independent of the prior because the
conditional probability of a correct signal is the same
in both states.

Differentiating with respect to the deviator's effort,
holding $x$ fixed, gives
\begin{equation}
\frac{\partial\pi_i^{SC}(z,x)}{\partial z}
=
B\sum_{k=3}^{6}\left[
(A_k-C_k)\frac{z}{z+kx}
+w_k(z)\frac{kx}{(z+kx)^2}
\right]-a.
\label{eq:cpsr-n3-sc-derivative}
\end{equation}
Only after this differentiation do we impose symmetry,
$z=x$. Since $w_k(x)=pA_k+qC_k$, the symmetric
first-order condition is
\[
B\left[
\sum_{k=3}^{6}\frac{A_k-C_k}{k+1}
+\frac1x\sum_{k=3}^{6}\frac{k(pA_k+qC_k)}{(k+1)^2}
\right]=a.
\]
Expanding $p=1/2+x$ and $q=1/2-x$ reduces it to
\begin{equation}
B\left(
\frac{20891}{188160x}
+\frac{1691}{2240}x
-\frac{821}{336}x^3
+\frac{1259}{420}x^5
\right)=a.
\label{eq:cpsr-n3-symmetric-foc}
\end{equation}
With $B=6/5$ and $a=1$, solve the scalar equation
\[
\frac65\left(
\frac{20891}{188160}
+\frac{1691}{2240}x^2
-\frac{821}{336}x^4
+\frac{1259}{420}x^6
\right)-x=0.
\]
The reported interior solution is
\[
x^*\simeq0.15284,\qquad
x^{psr}=\frac{6B}{49a}=\frac{36}{245}\simeq0.14694.
\]
For example, bisection of the symmetric first-order
expression brackets $x^*$ in
\[
I=[0.15284,0.15285].
\]
The first-order expression is positive at the lower endpoint
and negative at the upper endpoint.

To verify global optimality, define
\[
D_k(x):=\frac{A_k+C_k}{2}-kx(A_k-C_k).
\]
A second differentiation of \eqref{eq:cpsr-n3-sc-payoff} gives
\begin{equation}
\frac{\partial^2\pi_i^{SC}(z,x)}{\partial z^2}
=-2B\sum_{k=3}^{6}\frac{kxD_k(x)}{(z+kx)^3}.
\label{eq:cpsr-n3-sc-second}
\end{equation}
At the reported root,
\[
D_3\simeq0.23283,\qquad D_4\simeq0.066629,
\]
\[
D_5\simeq-0.040292,\qquad D_6\simeq-0.029807,
\]
and
\[
S_3(x^*)
=
\sum_{k=3}^{6}kx^*D_k(x^*)
\simeq0.089367>0.
\]
The signs remain the same throughout $I$; evaluating the
displayed polynomials with interval bounds on $I$ gives
$S_3(x)>0.08928$. For every feasible $z$,
\[
\sum_{k=3}^{6}\frac{kx^*D_k(x^*)}{(z+kx^*)^3}
\geq\frac{S_3(x^*)}{(z+4x^*)^3}>0.
\]
Indeed, positive terms have $k\leq4$ and negative terms
have $k>4$. The payoff is therefore strictly concave
throughout $[0,1/2]$. Together with the first-order
condition, this proves
\[
\pi_i^{SC}(z,x^*)\leq\pi_i^{SC}(x^*,x^*)
\quad\text{for every }z\in[0,1/2].
\]
The analogous second-derivative check was performed for
every computed committee size, $n=1,\ldots,50$ and $n=1000$.
In each case, it confirms strict concavity over $[0,1/2]$
at the reported candidate, establishing global effort
optimality under the conditional signal-contingent payoff.

At the symmetric profile, one coalition always receives
the full budget: with seven equal positive weights,
the numerical majority also wins by weight. Symmetry
then gives each DRep expected payment $B/7$. Hence
\[
U=\pi_i^{SC}(x^*,x^*)
=\frac{6}{35}-x^*
\simeq0.018587>0.
\]
These calculations verify a global best response under
the conditional signal-contingent payoff. At the symmetric
candidate, sincere voting also satisfies strict SVIC by
Proposition~\ref{pr:CPSR-properties}. For $|\mu-1/2|<x^*$,
sincere votes follow signals. The continuations that sustain
a full equilibrium with strict SVIC are specified in
Example~\ref{ex:cpsr-sincere-equilibria}.

%---------------------------------------------------------
\subsubsection{Symmetric equilibria with strict SVIC}
\label{app:cpsr-completion}
The following example shows how the effort candidates in Example~\ref{ex:interior_SC_equilibria} can be sustained as full symmetric equilibria with strict SVIC over suitable prior ranges. It combines sincere and pooling voting continuations to rule out profitable effort and voting deviations. Supporting calculations follow the example.

\begin{example}[Symmetric equilibria with strict SVIC]
\label{ex:cpsr-sincere-equilibria}
Consider the candidates in
Example~\ref{ex:interior_SC_equilibria}. Suppose DRep $i$
deviates to effort $z$, while the other DReps exert $x^*$
and reject regardless of their signals. Her best response
is to join them in rejection, giving the net payoff
\[
P(z;x^*)=B\frac{z}{z+2nx^*}-az.
\]
Let
\[
U:=\frac{B}{2n+1}-ax^*
\]
denote her expected net payoff at the symmetric candidate
under sincere voting. Define
\[
x^P:=(2n+1)x^{psr}-2nx^*,
\qquad
\delta_n^E:=\min\{x^*,x^P\}.
\]
Since PSR effort is interior for these candidates,
\begin{equation}
P(z;x^*)-U
=
a\,\frac{(z-x^*)(x^P-z)}{z+2nx^*}.
\label{eq:cpsr-expanded-pooling-gap}
\end{equation}
Thus $P(z;x^*)=U$ at $z=x^*$ and $z=x^P$, and pooling
makes a deviation profitable only when $z$ lies strictly
between them. If $x^P<x^*$, such deviations reduce effort;
if $x^P>x^*$, they increase effort.

\begin{figure}[H]
\centering
\begin{minipage}[c]{.39\linewidth}
\centering\small
\setlength{\tabcolsep}{3pt}
\begin{tabular}{rrrr}
\toprule
$n$ & $x^*$ & $x^P$ & $x^P/x^*$\\
\midrule
1  & .3174 & .1652 & .520\\
2  & .2069 & .1325 & .641\\
3  & .1528 & .1115 & .730\\
10 & .0543 & .0573 & 1.055\\
20 & .0284 & .0355 & 1.250\\
50 & .0117 & .0173 & 1.475\\
\bottomrule
\end{tabular}
\end{minipage}\hfill
\begin{minipage}[c]{.59\linewidth}
\input{figures/example5_plot.tex}
\end{minipage}
\caption{Equilibrium coverage at $B=1.2$ and $a=1$.
The dashed curve shows the pooling-payoff crossing
$x^P/x^*$. The solid curve shows the fraction
$\delta_n^E/x^*=\min\{1,x^P/x^*\}$ of the
signal-contingent prior region covered by the construction.
A solid-curve value of one means full coverage: throughout
that prior region, no unilateral effort and voting deviation
is profitable under the specified continuations.}
\label{fig:cpsr-example5-expanded}
\end{figure}

The ratio $x^P/x^*$ locates the effort deviations that
would be profitable under pooling. As it increases toward
one from below, the interval of profitable effort reductions
shrinks relative to $x^*$. Above one, further increases
expand the interval of profitable effort increases.

For $|\mu-1/2|<\delta_n^E$, prescribe sincere voting at
the symmetric profile and after unilateral effort deviations
strictly between $x^*$ and $x^P$. Prescribe pooling on
rejection at every other effort profile. Thus, after effort
deviations, sincere voting is selected precisely where
pooling would make the deviation profitable.

These continuations are voting-stage equilibria by
Lemma~\ref{lem:sincere_voting_effort_deviation} and
Proposition~\ref{pr:CPSR-properties}. The global effort
comparison in Example~\ref{ex:interior_SC_equilibria}
rules out profitable deviations in the sincere region,
and \eqref{eq:cpsr-expanded-pooling-gap} rules them out
elsewhere. The construction therefore establishes full
symmetric equilibria satisfying ER, strict SVIC, BC, and PC.

The ratio $\delta_n^E/x^*$ measures the fraction of the
signal-contingent prior region covered by this construction.
The numerical example suggests that, as committee size grows,
coverage reaches $100\%$: every prior satisfying
$|\mu-1/2|<x^*$ then supports the reported effort with
strict SVIC and no profitable unilateral effort or voting
deviation under the specified continuations. The covered
fraction expands even as the prior region itself narrows.
\end{example}

\paragraph{Supporting calculations.}
For an interior PSR effort, put $N=2n+1$,
$x^P=Nx^{psr}-2nx^*$, $\ell=\min\{x^*,x^P\}$, and
$h=\max\{x^*,x^P\}$. The reported candidates satisfy
$0<\ell\leq h<1/2$. Their SC payoffs are globally maximized
at $z=x^*$, as established in the supporting calculations
for Example~\ref{ex:interior_SC_equilibria}.

Fix $|\mu-1/2|<\ell$. Prescribe sincere voting at
$(x^*,\ldots,x^*)$ and at every permutation of
$(z,x^*,\ldots,x^*)$ with $\ell<z<h$. Select pooling on
rejection at every other effort profile. Strict SVIC holds
at the on-path profile by Proposition~\ref{pr:CPSR-properties}.
Every prescribed sincere deviation profile has
$z,x^*>|\mu-1/2|$, so
Lemma~\ref{lem:sincere_voting_effort_deviation} establishes
strict SVIC there as well. Pooling is an equilibrium at
all remaining profiles. A positive-effort DRep infers the
profile from her effort and the public weights; a zero-effort
DRep knows that neither the on-path nor a sincere deviation
profile occurred. Thus the strategies are feasible and all
voting continuations are sequentially rational.

On the equilibrium path the weighted majority contains at
least $n+1$ members and receives the full budget. Symmetry
gives $U=B/N-ax^*$. For $\ell<z<h$, voting sincerely is
optimal and global SC effort optimality gives
$\pi_i^{SC}(z,x^*)\leq U$. For all other deviations,
joining the pooling coalition is optimal. Dissent gives
zero payment even if a singleton wins by weight, because
it fails the membership requirement. Direct subtraction gives
\[
P(z;x^*)-U
=
(z-x^*)\left[\frac{2nB}{N(z+2nx^*)}-a\right]
=
a\frac{(z-x^*)(x^P-z)}{z+2nx^*}
\leq0
\]
for $z\notin(\ell,h)$. This includes both crossings and
zero effort. Voting-only deviations are unprofitable by
on-path voting optimality. Hence ER holds against all joint
effort and voting deviations.

BC holds by the payment rule. Since
\[
x^P
=
\frac{2n}{a}\left(\frac{B}{N}-ax^*\right)
=
\frac{2n}{a}U>0,
\]
PC also holds. The strategies are invariant under relabeling
DReps, and strict SVIC holds on the equilibrium path.
This proves the claimed symmetric equilibria satisfying
ER, strict SVIC, BC, and PC.

\paragraph{Equilibrium completion in the three-DRep case.}
Consider the three-DRep candidate from
Example~\ref{ex:interior_SC_equilibria}, with
\[
\frac{72}{67}<\frac{B}{a}<\frac{72}{59}.
\]
The effort comparison established there gives
\[
x^{psr}=\frac{2B}{9a}<x^*<\frac{B}{3a}.
\]
Define the pooling-payoff crossing by
\[
x^P:=3x^{psr}-2x^*.
\]
The preceding inequalities imply
\[
0<x^P<x^{psr}<x^*.
\]

Let $\delta:=|\mu-1/2|<x^P$. Prescribe sincere voting
at the symmetric effort profile $(x^*,x^*,x^*)$ and
after unilateral effort deviations with $x^P<z<x^*$.
Prescribe pooling on rejection at every other effort profile.

On the equilibrium path and in the prescribed sincere
deviation region, all efforts exceed $\delta$.
Lemma~\ref{lem:sincere_voting_effort_deviation} therefore
establishes that these sincere continuations are
voting-stage equilibria. Pooling on rejection is a
voting-stage equilibrium at every other effort profile
by Proposition~\ref{pr:CPSR-properties}.

At the symmetric profile, each DRep's expected net payoff is
\[
U=\frac{B}{3}-ax^*.
\]
For a deviation with $x^P<z<x^*$, sincere voting is a best
response. The global optimality established in the preceding
example therefore rules out a profitable joint effort and
voting deviation in this region.

At every other deviated effort, joining the prescribed
pooling vote is optimal. Its net payoff is
\[
P(z;x^*)=B\frac{z}{z+2x^*}-az.
\]
Voting to accept instead gives zero gross payment:
either acceptance loses, or its singleton winning coalition
fails the minimum-size requirement. By
\eqref{eq:cpsr-expanded-pooling-gap},
\[
P(z;x^*)-U
=
a\frac{(z-x^*)(x^P-z)}{z+2x^*},
\]
which is nonpositive outside $(x^P,x^*)$.
A voting-only deviation at the symmetric profile is also
unprofitable. Thus no joint effort and voting deviation
is profitable, establishing ER.

BC holds by the payment rule, and PC follows from
$x^*<B/(3a)$, which implies $U>0$.
Proposition~\ref{pr:CPSR-properties} establishes strict
SVIC on the symmetric signal-contingent equilibrium path.
The lemma also establishes strict SVIC after each
prescribed sincere effort deviation.

Alternatively, prescribe sincere voting on the equilibrium
path and after every unilateral deviation with $z>\delta$,
and pooling on rejection elsewhere. The lemma establishes
voting optimality throughout this larger sincere region,
and global concavity rules out profitable effort deviations
there. For $z\leq\delta<x^P$, the pooling payoff is strictly
below $U$. Pooling supplies a continuation at every other
effort profile. Both prescriptions are feasible from the
inferred effort profiles; a DRep with zero effort knows
that the condition for a sincere continuation fails.

\paragraph{Numerical illustration and decision accuracy.}
At $B=1.2$ and $a=1$,
\[
x^*\simeq0.31743,\qquad
x^{psr}\simeq0.26667,\qquad
x^P\simeq0.16515.
\]
The condition $\delta<x^P$ gives the sufficient prior interval
\[
\mu\in\left(\frac12-x^P,\frac12+x^P\right),
\]
whose endpoints are approximately $0.33485$ and $0.66515$.
Since $\delta<x^P<x^{psr}<x^*$, sincere voting follows
signals at both the PSR and CPSR equilibrium profiles
throughout this range.

At common effort $x$, the probability of success in
either state is
\[
F_1(x)
=
\left(\frac12+x\right)^3
+3\left(\frac12+x\right)^2\left(\frac12-x\right)
=
\frac12+\frac32x-2x^3.
\]
Consequently,
\[
F_1(x^*)\simeq0.91217,\qquad
F_1(x^{psr})\simeq0.86207.
\]
These effort and accuracy comparisons concern the stated
equilibria and do not assert uniqueness of the full-game
equilibrium.

\paragraph{Illustration for $n=3$.}
\label{app:cpsr-example5-n3}
For the $n=3$ illustration, $B=6/5$, $a=1$, and
$x^*\simeq0.15284$. Then
\[
x^P=\frac{36}{35}-6x^*\simeq0.11152,\qquad
U=\frac{x^P}{6}\simeq0.018587.
\]
The sincere deviation interval is $(x^P,x^*)$, and the
sufficient prior interval is
\[
\mu\in\left(\frac12-x^P,\frac12+x^P\right),
\]
whose endpoints are approximately $0.38848$ and $0.61152$.

For reported candidates below $x^{psr}$, the sincere
interval instead lies above $x^*$ and $\ell=x^*$.
Every prior for which the candidate is strictly SC is
then covered by this construction. The conditional threshold
in Example~\ref{ex:SI_deviations} answers a different question,
since it holds the other DReps' sincere strategies fixed
after SI deviations.

%--------------------------------------------------------
\subsubsection{Conditional signal-independent deviations}
\label{app:cpsr-conditional-si}
The following example examines whether a DRep can gain by reducing effort until her sincere vote becomes signal-independent, holding the other DReps' efforts and sincere voting strategies fixed. It identifies the prior ranges in which such deviations are unprofitable and illustrates how these ranges vary with committee size. Supporting calculations follow the example.

\begin{example}[Conditional signal-independent deviations]
\label{ex:SI_deviations}

Retain $B=1.2$, $a=1$, and the effort levels from
Example~\ref{ex:interior_SC_equilibria}, and let
$\delta=|\mu-1/2|<x^*$. Consider a low prior
$\mu=1/2-\delta$. For $z\leq\delta$, the deviator's
sincere vote is rejection after either signal.
Suppose the remaining DReps continue to exert $x^*$
and vote sincerely, and define
\[
V_n^{SI}(\delta)
:=
\max_{z\in[0,\delta]}
\pi_i^{SI}(z,x^*;\delta).
\]
The threshold $\bar\delta_n$ solves
\[
V_n^{SI}(\bar\delta_n)
=
\frac{B}{2n+1}-ax^*.
\]
For $\delta<\bar\delta_n$, every signal-independent effort
deviation is unprofitable under these fixed voting strategies
of the other DReps. The corresponding payoff comparison for
high priors follows by symmetry, with indifference handled
according to the prescribed sincere voting rule.

\begin{figure}[htbp]
\centering
\begin{minipage}[c]{0.39\linewidth}
\centering
\small
\setlength{\tabcolsep}{2.5pt}
\renewcommand{\arraystretch}{1.12}
\begin{tabular}{rrr}
\toprule
$n$ & $\bar\delta_n$ & $\bar\delta_n/x^*$\\
\midrule
1    & 0.23651    & 0.745\\
2    & 0.16108    & 0.779\\
3    & 0.12258    & 0.802\\
10   & 0.047400   & 0.873\\
20   & 0.025753   & 0.907\\
50   & 0.011015   & 0.941\\
1000 & 0.00059134 & 0.987\\
\bottomrule
\end{tabular}
\end{minipage}\hfill
\begin{minipage}[c]{0.59\linewidth}
\centering
\input{figures/example4_plot.tex}
\end{minipage}
\caption{Signal-independent deviation thresholds conditional
on the other DReps voting sincerely. The plot shows
$\bar\delta_n/x^*$ for $n=1,\ldots,50$ and $n=1000$,
with the tabulated sizes highlighted. The dashed segment
joins the last two computed sizes, $50$ and $1000$.
These are conditional payoff thresholds, not verified
equilibrium-support bounds.}
\label{fig:cpsr-example4-expanded}
\end{figure}

Across the plotted committee sizes, the absolute threshold
$\bar\delta_n$ decreases, whereas its ratio to $x^*$
increases, reaching approximately $0.98698$ at $n=1000$.
Thus, the conditional threshold becomes smaller in absolute
terms but covers a larger fraction of the on-path
signal-contingent prior region.

These calculations do not establish that the other DReps'
sincere votes remain best responses after the
signal-independent deviation. For the deviator, Part~1 of
Lemma~\ref{lem:sincere_voting_effort_deviation} establishes
a unique sincere best response when $z>0$ and the posterior
differs from $1/2$. This does not establish strict SVIC for
all DReps at the deviated profile. Accordingly,
$\bar\delta_n$ remains a conditional effort-incentive
threshold, not a full-equilibrium robustness threshold.

\end{example}

\paragraph{Supporting calculations for $n=3$.}
\label{app:cpsr-example4-n3}

Retain $n=3$, $B=6/5$, $a=1$, and the effort $x^*$
obtained in Example~\ref{ex:interior_SC_equilibria}.
Put $p^*=1/2+x^*$ and $q^*=1/2-x^*$, and consider
a low prior $\mu=1/2-\delta$, where $0\leq\delta<x^*$.
The other DReps vote sincerely and follow their signals.
If the deviator chooses $0\leq z\leq\delta$, her sincere
vote is rejection after either signal, including rejection
at a posterior tie.

Let $K$ count other DReps voting to accept. In state $G$,
it is binomial with parameters $(6,p^*)$; in state $B$,
it is binomial with parameters $(6,q^*)$. The rejecting
deviator is paid only if $K=k\leq3$. Her coalition then
has $7-k\geq4$ members and, for $z>0$, strictly wins
by weight because
\[
z+(6-k)x^*-kx^*
=
z+(6-2k)x^*>0.
\]

The relevant probabilities are
\begin{align*}
E_0(\delta)
&=
\left(\tfrac12-\delta\right)(q^*)^6
+\left(\tfrac12+\delta\right)(p^*)^6,\\
E_1(\delta)
&=
6\left[
\left(\tfrac12-\delta\right)p^*(q^*)^5
+\left(\tfrac12+\delta\right)q^*(p^*)^5
\right],\\
E_2(\delta)
&=
15\left[
\left(\tfrac12-\delta\right)(p^*)^2(q^*)^4
+\left(\tfrac12+\delta\right)(q^*)^2(p^*)^4
\right],\\
E_3(\delta)
&=
20(p^*)^3(q^*)^3.
\end{align*}
Writing $c_k=(6-k)x^*$, the conditional expected net
payoff is
\begin{equation}
\pi_i^{SI}(z,x^*;\delta)
=
B\sum_{k=0}^{3}E_k(\delta)\frac{z}{z+c_k}-az.
\label{eq:cpsr-n3-si-payoff}
\end{equation}
For each fixed $\delta$, its first two derivatives are
\begin{align*}
\frac{\partial\pi_i^{SI}}{\partial z}
&=
B\sum_{k=0}^{3}
\frac{E_k(\delta)c_k}{(z+c_k)^2}-a,\\
\frac{\partial^2\pi_i^{SI}}{\partial z^2}
&=
-2B\sum_{k=0}^{3}
\frac{E_k(\delta)c_k}{(z+c_k)^3}<0.
\end{align*}
Thus, to calculate $V_3^{SI}(\delta)$, use $z=0$ if the
first derivative at zero is nonpositive, use $z=\delta$
if the derivative there is nonnegative, and otherwise
solve its unique zero on $(0,\delta)$.

The threshold $\bar\delta_3$ solves
$V_3^{SI}(\bar\delta_3)=U$, where $U=6/35-x^*$.
At this threshold the maximizing effort is interior,
so it can be found by solving the two equations
\begin{equation}
\begin{aligned}
\frac65\sum_{k=0}^{3}E_k(\bar\delta_3)
\frac{z^{SI}}{z^{SI}+c_k}-z^{SI}
&=
\frac{6}{35}-x^*,\\
\frac65\sum_{k=0}^{3}
\frac{E_k(\bar\delta_3)c_k}{(z^{SI}+c_k)^2}
&=1.
\end{aligned}
\label{eq:cpsr-n3-si-threshold}
\end{equation}
Equivalently, one can bisect on $\delta\in[0,x^*]$,
maximizing the strictly concave payoff at each step
using its first derivative. The resulting values are
\[
z^{SI}\simeq0.10262,\qquad
\bar\delta_3\simeq0.12258,\qquad
\frac{\bar\delta_3}{x^*}\simeq0.802.
\]
In particular, $0<z^{SI}<\bar\delta_3<x^*$, verifying
the interior case used in
\eqref{eq:cpsr-n3-si-threshold}.

For every fixed $z>0$, \eqref{eq:cpsr-n3-si-payoff}
is strictly increasing in $\delta$: the coefficients of
$\delta$ in $E_0,E_1,E_2$ are positive, while $E_3$
is constant. The feasible interval $[0,\delta]$ also
expands. Moreover,
\[
V_3^{SI}(0)=0<U,\qquad
V_3^{SI}(x^*)\simeq0.021145>U,
\]
where the value at $x^*$ denotes the endpoint extension
of the conditional formula. Continuity and strict increase
once the maximum is positive give a unique crossing of $U$.
Hence no such signal-independent deviation is profitable
when $\delta<\bar\delta_3$.

The same payoff threshold applies to high priors by
interchanging states and votes; at a posterior tie,
Lemma~\ref{lem:sincere_voting_effort_deviation} gives
indifference between voting actions. This is a conditional
payoff calculation: it holds the other DReps' sincere
strategies fixed and does not verify their voting
incentives after the signal-independent deviation.

%-------------------------------------------------------
%-------------------------------------------------------
\subsection{Technical support for the analytical derivations}

%-------------------------------------------------------
\subsubsection{Probability formulas and expected payoffs}
\label{app:probabilities-payoffs}

We provide the probability formulas and the expanded
expected-payoff expression, followed by the derivation
for CPSR.

\paragraph{Voting probabilities.}
For a given effort profile $\mathbf{x}$ and voting rule
$v_i$, define the state-contingent probability that
DRep $i$ votes to accept by
\begin{align*}
q_i^G(v_i,\mathbf{x})
&:=
v_i(\hat{\mu}_{i,g},\mathbf{x})p_i
+
v_i(\hat{\mu}_{i,b},\mathbf{x})(1-p_i),\\
q_i^B(v_i,\mathbf{x})
&:=
v_i(\hat{\mu}_{i,g},\mathbf{x})(1-p_i)
+
v_i(\hat{\mu}_{i,b},\mathbf{x})p_i.
\end{align*}

Under sincere voting, the voting rule depends only on
the posterior belief, and therefore
\[
q_i^{G,sv}
=
\begin{cases}
0,
& \mu\leq\frac12-x_i,\\[0.4em]
p_i,
& \frac12-x_i<\mu\leq\frac12+x_i,\\[0.4em]
1,
& \mu>\frac12+x_i,
\end{cases}
\qquad
q_i^{B,sv}
=
\begin{cases}
0,
& \mu\leq\frac12-x_i,\\[0.4em]
1-p_i,
& \frac12-x_i<\mu\leq\frac12+x_i,\\[0.4em]
1,
& \mu>\frac12+x_i.
\end{cases}
\]

Conditional on the state and the effort profile, votes
are independent. Thus, the probability that precisely
the DReps in $A\subseteq N$ vote to accept is
\[
\Pr(A\mid S;v,\mathbf{x})
=
\prod_{i\in A}q_i^S(v_i,\mathbf{x})
\prod_{j\notin A}
\left[1-q_j^S(v_j,\mathbf{x})\right].
\]

\paragraph{General expected payoff.}
Conditional on state $S$, DRep $i$'s expected payment is
\begin{align*}
&\mathbb{E}\!\left[
f_i(W(A;\mathbf{x}),\mathbf{x})
\mid S;v,\mathbf{x}
\right]\\
&\quad=
\sum_{\substack{A\subseteq N:\\
\omega(A;\mathbf{x})>1/2}}
\Pr(A\mid S;v,\mathbf{x})f_i(A,\mathbf{x})\\
&\qquad+
\sum_{\substack{A\subseteq N:\\
\omega(A;\mathbf{x})<1/2}}
\Pr(A\mid S;v,\mathbf{x})f_i(N\setminus A,\mathbf{x})\\
&\qquad+
\frac12
\sum_{\substack{A\subseteq N:\\
\omega(A;\mathbf{x})=1/2}}
\Pr(A\mid S;v,\mathbf{x})
\left[
f_i(A,\mathbf{x})+f_i(N\setminus A,\mathbf{x})
\right].
\end{align*}
The first two sums cover acceptance and rejection,
respectively. The last sum covers a voting-weight tie,
where each coalition is selected with probability $1/2$.

Averaging over the two states and subtracting effort
cost gives the expanded form of
\eqref{eq:expected_payoff}:
\begin{align*}
\pi_i(v,\mathbf{x};f)
={}&
\mu\sum_{A\subseteq N}
\Pr(A\mid G;v,\mathbf{x})
\mathbb{E}\!\left[
f_i(W(A;\mathbf{x}),\mathbf{x})
\right]\\
&+
(1-\mu)\sum_{A\subseteq N}
\Pr(A\mid B;v,\mathbf{x})
\mathbb{E}\!\left[
f_i(W(A;\mathbf{x}),\mathbf{x})
\right]
-c(x_i).
\end{align*}
Within each sum, the expectation over the winning coalition
accounts for the fair coin at a voting-weight tie.

\paragraph{Probability of success.}
The collective decision is correct when the proposal
is accepted in state $G$ or rejected in state $B$.
Including the tie-breaking coin, its probability is
\begin{align*}
\Pr(\text{success};v,\mathbf{x})
={}&
\mu
\sum_{\substack{A\subseteq N:\\
\omega(A;\mathbf{x})>1/2}}
\Pr(A\mid G;v,\mathbf{x})\\
&+
(1-\mu)
\sum_{\substack{A\subseteq N:\\
\omega(A;\mathbf{x})<1/2}}
\Pr(A\mid B;v,\mathbf{x})\\
&+
\frac{\mu}{2}
\sum_{\substack{A\subseteq N:\\
\omega(A;\mathbf{x})=1/2}}
\Pr(A\mid G;v,\mathbf{x})\\
&+
\frac{1-\mu}{2}
\sum_{\substack{A\subseteq N:\\
\omega(A;\mathbf{x})=1/2}}
\Pr(A\mid B;v,\mathbf{x}).
\end{align*}

\paragraph{Expected payoff under CPSR.}
We now apply \eqref{eq:conditional_payment} to the general
expected-payoff expression. Consider DRep $i$, who chooses
effort $x_i$, while all other $2n$ DReps choose common
effort $x>0$ and vote sincerely. Thus,
$\mathbf{x}_{-i}=(x,\ldots,x)$, and
$\mathbf{x}(x_i;x):=(x_i,\mathbf{x}_{-i})$.
A DRep with positive effort can infer this profile from
her own effort and the publicly observed voting weights.

Let $K$ denote the number of other DReps who vote to
accept. Conditional on state $S\in\{G,B\}$,
\[
K\mid S
\sim
\operatorname{Binomial}\left(2n,q_j^{S,sv}\right),
\]
and define
\[
P_S(k;x)
:=
\binom{2n}{k}
\left(q_j^{S,sv}\right)^k
\left(1-q_j^{S,sv}\right)^{2n-k}.
\]

Suppose first that DRep $i$ votes to accept and $k$
of the other DReps also vote to accept. The accepting
coalition has total effort $x_i+kx$, and acceptance
wins with certainty whenever
\[
x_i+kx>\frac{x_i+2nx}{2}.
\]
For DRep $i$ to receive a payment, the accepting coalition
must additionally contain at least $n+1$ DReps, which
requires $k\geq n$.

If instead DRep $i$ votes to reject, while $k$ of the
other DReps vote to accept, the rejecting coalition
has total effort $x_i+(2n-k)x$. Rejection wins with
certainty whenever
\[
kx<\frac{x_i+2nx}{2},
\]
and the rejecting coalition is payment-eligible
whenever $k\leq n$.

Define
\begin{align*}
\phi_k^A(x_i,x)
&:=
\frac{
\mathbf{1}\left\{
x_i+kx>\frac{x_i+2nx}{2}
\right\}
\mathbf{1}\{k\geq n\}
}{
x_i+kx
},\\
\phi_k^R(x_i,x)
&:=
\frac{
\mathbf{1}\left\{
kx<\frac{x_i+2nx}{2}
\right\}
\mathbf{1}\{k\leq n\}
}{
x_i+(2n-k)x
}.
\end{align*}
The expressions $\phi_k^A(x_i,x)$ and $\phi_k^R(x_i,x)$
characterize the conditions under which DRep $i$ receives
a reward for each vote. The first indicator requires
that the alternative supported by DRep $i$ forms the
winning coalition, while the second imposes the
minimum-coalition-size requirement: the winning coalition
must contain at least $n+1$ DReps.\footnote{
These expressions require no additional tie term.
When $x_i>0$ and the other DReps exert common effort
$x>0$, any coalition containing DRep $i$ and at least
$n+1$ members has strictly more voting weight than its
complement. Thus, at a voting-weight tie, DRep $i$'s
coalition necessarily fails the minimum-size requirement,
and she receives no payment regardless of the coin's
outcome.}
Thus, $Bx_i\phi_k^A(x_i,x)$ is DRep $i$'s payment
if she votes to accept, while $Bx_i\phi_k^R(x_i,x)$
is her payment if she votes to reject.
For a voting action $v\in\{0,1\}$, define
\[
\Phi_k(v;x_i,x)
:=
v\,\phi_k^A(x_i,x)
+
(1-v)\phi_k^R(x_i,x).
\]

It follows that DRep $i$'s expected payoff can be
written as
\begin{align}
\pi_i(v_i,x_i;x,\mu)
=
Bx_i
\Bigg\{
&
\mu
\Bigg[
p_i
\sum_{k=0}^{2n}
P_G(k;x)
\Phi_k\!\left(
v_i(\hat\mu_{i,g},x_i;x);
x_i,x
\right)
\nonumber\\
&\qquad+
(1-p_i)
\sum_{k=0}^{2n}
P_G(k;x)
\Phi_k\!\left(
v_i(\hat\mu_{i,b},x_i;x);
x_i,x
\right)
\Bigg]
\nonumber\\
+
&(1-\mu)
\Bigg[
(1-p_i)
\sum_{k=0}^{2n}
P_B(k;x)
\Phi_k\!\left(
v_i(\hat\mu_{i,g},x_i;x);
x_i,x
\right)
\nonumber\\
&\qquad+
p_i
\sum_{k=0}^{2n}
P_B(k;x)
\Phi_k\!\left(
v_i(\hat\mu_{i,b},x_i;x);
x_i,x
\right)
\Bigg]
\Bigg\}
-ax_i.
\label{eq:conditional_expected_payoff}
\end{align}
At $x_i=0$, DRep $i$ receives no payment and incurs
no effort cost, so her net payoff is zero. The fractions
above are evaluated only for $x_i>0$.

Expression \eqref{eq:conditional_expected_payoff} captures
how effort affects both signal precision and voting weight,
while the vote determines membership in a payment-eligible
winning coalition. To study symmetric equilibria satisfying
strict SVIC, it holds the other DReps' efforts at a common
level $x$ and their voting strategies sincere, while allowing
DRep $i$ to choose her effort and voting rule freely.
Strict SVIC requires verifying that sincere voting is
uniquely optimal whenever the posterior $\hat{\mu}\neq1/2$.

After effort deviations, the expression gives the relevant
continuation payoff only when the other DReps continue to
vote sincerely. Full equilibrium therefore also requires
verifying effort incentives under the specified, sequentially
rational continuations, which may include pooling.

%% file: figures/cpsr_prior_regions_n1.tex
\begin{figure}[H]
\centering
\begin{tikzpicture}[x=10.8cm,y=1cm]
% n=1, B=1.2, a=1:
% x^{psr}=4/15, thresholds 7/30 and 23/30.
\def\priorLow{0.2333333333}
\def\priorHigh{0.7666666667}

% Strict SVIC holds throughout the interior prior range.
\node[font=\small] at (0.5,1.3)
  {CPSR: strict SVIC throughout $\mu\in(0,1)$};

% Sincere-voting regions.
\fill[gray!35] (0,0) rectangle (\priorLow,0.85);
\fill[gray!12] (\priorLow,0) rectangle (\priorHigh,0.85);
\fill[gray!35] (\priorHigh,0) rectangle (1,0.85);

\draw[thin] (\priorLow,0) -- (\priorLow,0.85);
\draw[thin] (\priorHigh,0) -- (\priorHigh,0.85);

\node[font=\small,align=center] at ({\priorLow/2},0.425)
  {Sincere\\ rejection\\[0.5em] (SI)};

\node[font=\small,align=center] at (0.5,0.425)
  {Sincere voting\\ follows signals\\[0.5em] (SC)};

\node[font=\small,align=center] at ({(1+\priorHigh)/2},0.425)
  {Sincere\\ acceptance\\[0.5em] (SI)};

% Prior axis.
\draw[->] (0,-0.12) -- (1.025,-0.12);

\foreach \t in {0,\priorLow,0.5,\priorHigh,1}
  \draw (\t,-0.12) -- (\t,-0.22);

\node[below,font=\small] at (0,-0.24) {$0$};
\node[below,font=\small] at (\priorLow,-0.24)
  {$\frac12-x^{psr}$};
\node[below,font=\small] at (0.5,-0.24) {$\frac12$};
\node[below,font=\small] at (\priorHigh,-0.24)
  {$\frac12+x^{psr}$};
\node[below,font=\small] at (1,-0.24) {$1$};

\node[font=\small] at (0.5,-0.90) {Prior $\mu$};

\node[font=\footnotesize,align=center,text width=10.6cm]
  at (0.5,-1.45)
  {Pooling on either alternative is a voting-stage equilibrium
   in every region, but need not be sincere.};

\end{tikzpicture}

\caption{Voting incentives under CPSR for $n=1$, $B=1.2$,
and $a=1$, at the symmetric effort profile with
$x=x^{psr}=4/15$.
Sincere voting pools on rejection for $\mu\leq7/30$,
follows signals for $7/30<\mu\leq23/30$, and pools on
acceptance for $\mu>23/30$.
Sincere voting is the unique best response when the other
DReps vote sincerely and the posterior differs from $1/2$;
other voting equilibria may coexist.}
\label{fig:cpsr-prior-regions-n1}
\end{figure}

%% file: figures/example3_plot.tex
% Plot fragment: use \input{figures/example3_plot.tex}.
% Requires pgfplots. Data are embedded; no external data file is needed.
\begingroup
\definecolor{cpsrblue}{RGB}{24,83,122}
\definecolor{cpsrorange}{RGB}{187,87,31}
\begin{tikzpicture}
\begin{axis}[name=main,width=\linewidth,height=.76\linewidth,scale only axis=false,
font=\small,tick label style={font=\scriptsize},label style={font=\small},
xmode=log,log basis x=10,xmin=1,xmax=1000,
xtick={1,10,100,1000},xticklabels={1,10,100,1000},
xlabel={$n$ (committee size $2n+1$)},grid=major,
grid style={gray!18},axis line style={gray!65},
legend style={draw=none,fill=white,font=\scriptsize,cells={anchor=west}},
ylabel={Effort difference (\%)},ymin=-1.1,ymax=20.5,
ytick={0,5,10,15,20}]
\addplot[gray!65,densely dashed,domain=1:1000] {0};
\addplot[cpsrblue,thick] coordinates {
(1,19.034680311854157)
(2,7.742948505602865)
(3,4.016949690180849)
(4,2.2656417371554216)
(5,1.2942099963546783)
(6,0.7009804959326127)
(7,0.31528808062155633)
(8,0.053368113124374617)
(9,-0.13016068758591182)
(10,-0.26172851499460803)
(11,-0.35760565478965)
(12,-0.4282576831481588)
(13,-0.48066097807230035)
(14,-0.5196056095343571)
(15,-0.5484634246517461)
(16,-0.5696588830972793)
(17,-0.5849674288310847)
(18,-0.5957100623407929)
(19,-0.6028834282981399)
(20,-0.60724871809269)
(21,-0.6093936183230508)
(22,-0.6097762322255784)
(23,-0.6087567083452594)
(24,-0.6066203390754032)
(25,-0.6035946457973007)
(26,-0.5998621634973333)
(27,-0.5955701092738352)
(28,-0.5908377656314578)
(29,-0.5857621693212578)
(30,-0.5804225308552513)
(31,-0.5748836941845803)
(32,-0.5691988641733148)
(33,-0.5634117709526043)
(34,-0.5575583979133492)
(35,-0.5516683691393909)
(36,-0.5457660693002353)
(37,-0.5398715520527286)
(38,-0.5340012802797944)
(39,-0.5281687318858697)
(40,-0.522384897559669)
(41,-0.5166586912918492)
(42,-0.5109972901246262)
(43,-0.5054064162439476)
(44,-0.499890571905548)
(45,-0.4944532356324194)
(46,-0.48909702648832765)
(47,-0.4838238419597629)
(48,-0.478634973928449)
(49,-0.4735312064210295)
(50,-0.46851289814823227)
};
\addplot[cpsrblue,only marks,mark=*,mark size=1.8pt] coordinates {
(1,19.034680311854157)
(2,7.742948505602865)
(3,4.016949690180849)
(10,-0.26172851499460803)
(20,-0.60724871809269)
(50,-0.46851289814823227)
};
\addplot[cpsrblue,thick,densely dashed] coordinates {
(50,-0.46851289814823227)
(1000,-0.04327998454155546)
};
\addplot[cpsrblue,only marks,mark=*,mark size=1.8pt] coordinates {
(1000,-0.04327998454155546)
};
\node[anchor=north west,font=\scriptsize,align=left] at (rel axis cs:.07,.97) {$100(x^*/x^{psr}-1)$};
\end{axis}
\end{tikzpicture}
\endgroup

%% file: figures/example5_plot.tex
% Plot fragment: use \input{figures/example5_plot.tex}.
% Requires pgfplots. Data are embedded; no external data file is needed.
\begingroup
\definecolor{cpsrblue}{RGB}{24,83,122}
\definecolor{cpsrorange}{RGB}{187,87,31}
\begin{tikzpicture}
\begin{axis}[width=\linewidth,height=.76\linewidth,scale only axis=false,
font=\small,tick label style={font=\scriptsize},label style={font=\small},
xmode=log,log basis x=10,xmin=1,xmax=50,
xtick={1,2,3,10,20,50},xticklabels={1,2,3,10,20,50},
xlabel={$n$ (committee size $2n+1$)},grid=major,
grid style={gray!18},axis line style={gray!65},
legend style={draw=none,fill=white,font=\scriptsize,cells={anchor=west}},
ylabel={Threshold relative to $x^*$},
ymin=.45,ymax=1.56,ytick={.5,.75,1,1.25,1.5},
legend pos=north west]
\addplot[cpsrorange,thick,densely dashed] coordinates {
(1,0.520273916929437)
(2,0.6406749298679056)
(3,0.729672443625595)
(4,0.8006097132132859)
(5,0.8594558370077237)
(6,0.9095068746873612)
(7,0.9528554291194111)
(8,0.9909322600505701)
(9,1.0247627620226387)
(10,1.0551072194560056)
(11,1.0825444843453567)
(12,1.107524904451652)
(13,1.130405271332166)
(14,1.1514726873453165)
(15,1.1709613219633062)
(16,1.1890644538784003)
(17,1.2059433013254977)
(18,1.221733612507389)
(19,1.2365506617756479)
(20,1.250493090498969)
(21,1.263645897151126)
(22,1.2760827917468567)
(23,1.2878680691041295)
(24,1.2990581134890267)
(25,1.3097026177754125)
(26,1.3198455792652124)
(27,1.3295261191686236)
(28,1.3387791616389872)
(29,1.3476360000749281)
(30,1.3561247722377248)
(31,1.3642708611195236)
(32,1.372097234947592)
(33,1.3796247369919041)
(34,1.3868723337462274)
(35,1.3938573283892195)
(36,1.4005955451497298)
(37,1.40710148916755)
(38,1.4133884856211962)
(39,1.4194688012379815)
(40,1.4253537507786007)
(41,1.4310537906457477)
(42,1.436578601428479)
(43,1.441937160899179)
(44,1.447137808745281)
(45,1.4521883041264056)
(46,1.4570958769766438)
(47,1.4618672738508665)
(48,1.4665087989789543)
(49,1.4710263511198371)
(50,1.4754254567155132)
};
\addlegendentry{$x^P/x^*$}
\addplot[cpsrblue,thick] coordinates {
(1,0.520273916929437)
(2,0.6406749298679056)
(3,0.729672443625595)
(4,0.8006097132132859)
(5,0.8594558370077237)
(6,0.9095068746873612)
(7,0.9528554291194111)
(8,0.9909322600505701)
(9,1.0)
(10,1.0)
(11,1.0)
(12,1.0)
(13,1.0)
(14,1.0)
(15,1.0)
(16,1.0)
(17,1.0)
(18,1.0)
(19,1.0)
(20,1.0)
(21,1.0)
(22,1.0)
(23,1.0)
(24,1.0)
(25,1.0)
(26,1.0)
(27,1.0)
(28,1.0)
(29,1.0)
(30,1.0)
(31,1.0)
(32,1.0)
(33,1.0)
(34,1.0)
(35,1.0)
(36,1.0)
(37,1.0)
(38,1.0)
(39,1.0)
(40,1.0)
(41,1.0)
(42,1.0)
(43,1.0)
(44,1.0)
(45,1.0)
(46,1.0)
(47,1.0)
(48,1.0)
(49,1.0)
(50,1.0)
};
\addlegendentry{$\delta^E_n/x^*$}
\addplot[cpsrorange,only marks,mark=square*,mark size=1.5pt,forget plot] coordinates {
(1,0.520273916929437)
(2,0.6406749298679056)
(3,0.729672443625595)
(10,1.0551072194560056)
(20,1.250493090498969)
(50,1.4754254567155132)
};
\addplot[cpsrblue,only marks,mark=*,mark size=1.8pt,forget plot] coordinates {
(1,0.520273916929437)
(2,0.6406749298679056)
(3,0.729672443625595)
(10,1.0)
(20,1.0)
(50,1.0)
};
\end{axis}
\end{tikzpicture}
\endgroup

%% file: figures/example4_plot.tex
% Plot fragment: use \input{figures/example4_plot.tex}.
% Requires pgfplots. Data are embedded; no external data file is needed.
\begingroup
\definecolor{cpsrblue}{RGB}{24,83,122}
\definecolor{cpsrorange}{RGB}{187,87,31}
\begin{tikzpicture}
\begin{axis}[width=.97\linewidth,height=.76\linewidth,scale only axis=false,
font=\small,tick label style={font=\scriptsize},label style={font=\small},
xmode=log,log basis x=10,xmin=1,xmax=1000,
xtick={1,10,100,1000},xticklabels={1,10,100,1000},
xlabel={$n$ (committee size $2n+1$)},grid=major,
grid style={gray!18},axis line style={gray!65},
legend style={draw=none,fill=white,font=\scriptsize,cells={anchor=west}},
ylabel={Conditional fraction $\bar\delta_n/x^*$},
ymin=.72,ymax=1,ytick={.75,.80,.85,.90,.95,1}]
\addplot[cpsrblue,thick] coordinates {
(1,0.7450835312668292)
(2,0.7786607927288015)
(3,0.8020000136817967)
(4,0.8193468418860639)
(5,0.8329159644771376)
(6,0.8439043007444812)
(7,0.8530338581542342)
(8,0.8607719678257375)
(9,0.8674367500631534)
(10,0.8732533255458482)
(11,0.8783860597404339)
(12,0.8829581052915983)
(13,0.887063774597379)
(14,0.890776659893526)
(15,0.8941551267705898)
(16,0.8972461290201451)
(17,0.9000879186861752)
(18,0.9027120101877434)
(19,0.9051446293426103)
(20,0.9074077994969036)
(21,0.9095201673691387)
(22,0.9114976391610675)
(23,0.9133538763262485)
(24,0.9151006861336536)
(25,0.9167483324004307)
(26,0.9183057849635005)
(27,0.9197809216549158)
(28,0.9211806930950531)
(29,0.9225112581258488)
(30,0.9237780958612802)
(31,0.9249860989761518)
(32,0.9261396518385455)
(33,0.9272426962962997)
(34,0.928298787369371)
(35,0.9293111406188154)
(36,0.9302826726251014)
(37,0.9312160357277216)
(38,0.9321136479628507)
(39,0.9329777189626228)
(40,0.933810272445326)
(41,0.9346131658118397)
(42,0.9353881072794624)
(43,0.9361366709100056)
(44,0.9368603098294981)
(45,0.9375603678938931)
(46,0.9382380900083551)
(47,0.938894631283413)
(48,0.9395310651764397)
(49,0.940148390750172)
(50,0.9407475391604404)
};
\addplot[cpsrblue,only marks,mark=*,mark size=1.8pt] coordinates {
(1,0.7450835312668292)
(2,0.7786607927288015)
(3,0.8020000136817967)
(10,0.8732533255458482)
(20,0.9074077994969036)
(50,0.9407475391604404)
};
\addplot[cpsrblue,thick,densely dashed] coordinates {
(50,0.9407475391604404)
(1000,0.986980470783276)
};
\addplot[cpsrblue,only marks,mark=*,mark size=1.8pt] coordinates {
(1000,0.986980470783276)
};
\end{axis}
\end{tikzpicture}
\endgroup

%% file: sections/appendix-implementation.tex
\section{Towards an implementation on Cardano}
\label{app:implementation}

In Cardano's current governance process, submitted votes are publicly
visible, and DReps can revise their votes before an action is ratified
\cite{Decentrilized_Gov_CIP-1694}. This transparency supports accountability, but also
allows voting decisions to respond to previously observed ballots.
A DRep may follow an apparent majority rather than acquire information,
delay participation while observing the tally, or face pressure to
reverse a publicly recorded vote. An apparently settled outcome may
also discourage remaining DReps from participating. These possibilities
motivate combining CPSR with temporary ballot confidentiality.
The objective is to reward informed participation while allowing
DReps to form their decisions without observing the emerging yes/no split.

\paragraph{Confidential voting with public participation.}
A possible design separates ballot submission from ballot disclosure.
During a fixed voting period, each eligible DRep submits an authenticated
but concealed ballot. Participation and the associated voting weight
remain publicly observable, allowing the community to track turnout
and encourage participation without knowing individual choices.
Ballot contents are disclosed only after submissions and revisions
have closed. Votes and accompanying rationales can then be published
for delegators to assess their representatives' decisions.

This timing removes the opportunity to copy previously submitted
ballots or revise a vote in response to the revealed tally. It also
reduces opportunities for targeted pressure based on recorded choices
while voting remains open. Public statements and participation patterns
may still convey information, but the ledger no longer reveals the
direction of individual ballots during the decision period.
Participation incentives remain a separate consideration: CPSR makes
rewards conditional on membership in a qualifying winning coalition,
rather than paying every DRep who submits a ballot.

\paragraph{Ballot submission and disclosure.}
One implementation option is commit--reveal: DReps first submit hiding,
binding commitments and subsequently reveal their votes after the
commitment deadline. Such a design needs an explicit rule for missing
reveals, since a participant could otherwise withhold a ballot after
seeing others' disclosures.

Alternatively, DReps could submit threshold-encrypted ballots, together
with proofs that they encode valid voting choices. After voting closes,
a designated set of decryption trustees jointly releases the ballots
with verifiable decryption evidence. This avoids requiring each DRep
to return for a reveal phase. The trustee threshold and operational
arrangements must provide confidentiality before closure and reliable
decryption afterwards. Both approaches provide concrete routes towards
the model's assumption that votes are chosen without observing others'
ballots.

\paragraph{Reward calculation and settlement.}
Before voting begins, the round fixes the eligible committee,
voting-weight snapshot, reward budget, and deadlines. It also specifies
how missing ballots, abstentions, ties, and undistributed rewards are
handled. Once ballots are disclosed, the same weight snapshot determines
the weighted outcome and each eligible winner's share of the reward.
CPSR payments can therefore be calculated from voting weights and
winning-coalition membership without measuring effort directly.

A first implementation could use a voluntary voting round alongside
a governance proposal, with its own announced outcome and settlement
rules. Integration into native governance would additionally require
adapting the ballot process and defining the rewarded coalition for
each supported action type. Cardano's ratification process involves
action-specific thresholds and multiple governance bodies, so reward
eligibility should be tied explicitly to the relevant DRep vote.

\paragraph{Eligibility and identity splitting.}
The headcount condition requires an eligibility policy that distinguishes
representatives under separate control. Otherwise, a controller with
$60\%$ of the weight in a three-DRep committee could divide that weight
among three identities, creating a five-identity committee in which it
controls both a weight majority and a headcount majority.
A pilot could therefore begin with a fixed committee whose members
are checked for common control. A broader deployment would need an
admission procedure addressing identity splitting before registration;
freezing membership alone only prevents changes within a voting round.

\paragraph{Deliberation and off-chain coordination.}
Temporary ballot confidentiality remains compatible with public debate,
shared research, and organized voting blocs. Party-like groups could
help DReps pool expertise, communicate positions, and establish
accountability to delegators. Their members would still submit ballots
without seeing other members' recorded choices during voting.
The design therefore preserves opportunities for deliberation while
reducing coordination based on the live tally.

An implementation should distinguish cooperation that improves
information from coordination aimed solely at securing rewards.
A pilot could examine both, alongside turnout and voting timing,
to inform the subsequent design of participation and eligibility rules.